\documentclass{article}

\usepackage[english]{babel}
\usepackage{amsthm}
\usepackage[letterpaper,top=2cm,bottom=2cm,left=3cm,right=3cm,marginparwidth=1.75cm]{geometry}

\usepackage{adjustbox}
\usepackage{lipsum}
\usepackage{amsmath}
\usepackage{graphicx}
\usepackage[colorlinks=true, allcolors=blue]{hyperref}
\theoremstyle{definition}

\usepackage{authblk}
\theoremstyle{remark}

\newtheorem{theorem}{\protect\textbf{Theorem}}
\newtheorem{proposition}{\protect\textbf{Proposition}}
\newtheorem{corollary}{\protect\textbf{Corollary}}

\newtheorem{assumption}{\protect\textbf{Assumption}}
\usepackage{multirow}
\usepackage[table,xcdraw]{xcolor}
\usepackage{setspace}
\usepackage{float}
\providecommand{\keywords}[1]
{
  \small	
  {\textit{$~~~~~~$ KEYWORDS:}} #1
}

\title{Testing the Exogeneity of Instrumental Variables and Regressors in Linear Regression Models Using Copulas}
\author{Seyed Morteza Emadi}
\affil{University of North Carolina at Chapel-Hill \\ Kenan-Flagler Business School}
\date{}
\begin{document}
\maketitle

\begin{abstract}
We provide a Copula-based approach to test the exogeneity of instrumental variables in linear regression models. We show that the exogeneity of instrumental variables is equivalent to the exogeneity of their standard normal transformations with the same CDF value. Then, we establish a Wald test for the exogeneity of the instrumental variables. We demonstrate the performance of our test using simulation studies. Our simulations show that if the instruments are actually endogenous, our test rejects the exogeneity hypothesis approximately 93\% of the time at the 5\% significance level. Conversely, when instruments are truly exogenous, it dismisses the exogeneity assumption less than 30\% of the time on average for data with 200 observations and less than 2\%  of the time for data with 1,000 observations. Our results demonstrate our test's effectiveness, offering significant value to applied econometricians.
\end{abstract}

\section{Introduction}

Suppose that we want to estimate the following linear regression model: $ Y_t=\beta_0+\beta^T X_t+\alpha P_t+\epsilon_t,~t=\{1,...,T\},$ where $X_t$ is the $k \times 1$ vector of exogenous regressors that are uncorrelated with the error term $\epsilon_t$, while $P_t$ is a scalar variable that is endogenous; i.e., it is potentially correlated with the error term. We are particularly interested in estimating the coefficient of the endogenous variable $\alpha$. Using instrumental variables is a popular approach to address this endogeneity issue (\cite{angrist1996identification}). Suppose that we have access to $m \times 1$ vector of instrumental variables denoted by $Z_t$, where $Z_t$ does not include any variable in $X_t$. There are two fundamental conditions in the approach of instrumental variables: first, the instruments $Z_t$ should be correlated with the exogenous variable (\textit{relevance}), and second, the instruments $Z_t$ should not be correlated with the error term (\textit{exogeneity} or \textit{validity}). Evaluation of the relevance condition is relatively straightforward, but testing the instrument exogeneity condition is not (\cite{wooldridge2010econometric}). Usually, researchers justify this condition using economic-theoretical arguments. However, in many situations, these arguments could be considered subjective beliefs that could not be supported by the data. On the other hand, it has been shown in the existing literature (\cite{kiviet2020testing}, \cite{dufour2003identification} and \cite{bound1995problems}) that a violation of this non-testable condition can not only lead to loss of efficiency, but can also result in even a higher estimation bias in comparison to the regular OLS. 

To alleviate these concerns, there have been two streams of literature. The first stream does not address testing the instrument exogeneity condition. It explores the effect of violation of this condition on the estimate of interest and the resulting inference (\cite{ashley2009assessing}, \cite{kraay2012instrumental} and \cite{nevo2012identification}). In other words, these works assess the robustness of different inferential conclusions to uncertainty in the validity of instruments. The second stream attempts to test the exogeneity of instruments with some caveats. \cite{sargan1958estimation} and \cite{hansen1982large} study formal testing of over-identification restrictions contingent on the validity of the initial just-identifying set of instruments. In other words, these works' methods rely on a non-testable assumption that the initial set of just-identifying instruments satisfies the exogeneity condition. As a result, as illustrated in \cite{deaton2010instruments} and \cite{parente2012cautionary}, these overidentification tests cannot reliably test the exogeneity for all instruments. 

Taking a different direction, a more recent stream of work evaluates the instrument exogeneity condition using an indirect approach (\cite{kiviet2020testing} and \cite{kripfganz2021kinkyreg}). In \cite{kiviet2020testing}, the authors adopt nonorthogonal moment conditions in regression by developing an instrument-free approach to estimate the parameters. The authors assume that the correlation between the endogenous variable and the error term is known. Then, they show that the instrument exogeneity condition is satisfied if the correlation value between the endogenous variable and the error term, which rejects the instrument exogeneity condition, is not in an admissible range. However, the issue is that the admissible range for the correlation coefficient of the endogenous variable and the error term is unknown from the data. Therefore, a researcher must rely on expert knowledge of the admissible degree of endogeneity. Relying on expert knowledge induces subjectivity to the approach in \cite{kiviet2020testing}. In \cite{kitagawa2015test}, the author tests a necessary condition for the exogeneity of the instruments based on ordering the conditional joint distribution of the outcome and the treatment depending on the instrument. However, the method in \cite{kitagawa2015test} applies only to cases with binary treatments and discrete instruments, which limits its applicability.  

Despite the lack of reliable methods to test instrument exogeneity in the existing literature, the issue remains critical due to the vast number of papers that leverage the instrumental variable approach. To address this issue, we propose a Copula-based method to test the exogeneity of instruments from the error term. Copulas have been used as an instrument-free approach to address endogeneity issues in the estimation of linear regression models (\cite{park2012handling} and \cite{yang2022addressing}). However, our work is the first to use this approach to test the exogeneity of instruments from the error term.

To model the endogeneity of $P_t$, we consider the reduced-form equation for it given by $ P_t=\delta_0+\delta^T X_t+\gamma^T Z_{t}+\eta_t$. Given that $X_t$ is exogenous, the endogeneity of $P_t$ implies the potential endogeneity of $Z_t$ and the error term $\eta_t$. Therefore, we model the endogeneity of $P_t$ by characterizing the joint distribution of the error term in the original regression equation $\epsilon_t$, $Z_t$ and $\eta_t$ by a Gaussian copula. The Gaussian copula is a multivariate normal distribution on the normal transformations of variables. The normal transformation of a variable in $\epsilon_t$, $Z_t$, and $\eta_t$ is a variable from the standard normal distribution with the same CDF value.  

\begin{sloppypar}
We assume that the error term has a normal distribution, which is a widely used assumption in the literature (\cite{kleibergen2003bayesian}, \cite{ebbes2005solving}, \cite{rossi2003bayesian}, \cite{park2012handling} and \cite{yang2022addressing}). However, our simulation studies show that our approach is also robust to non-normal error terms. Following the normality assumption of the error term, we show that the exogeneity of any variable from the error term $\epsilon_t$ is equivalent to the exogeniety of its normal transformation from the error term. This is not a trivial result because the endogeneity issue is the lack of Pearson correlation, which is different from statistical independence. Therefore, testing for the exogeneity of a variable is tantamount to verifying whether the correlation between the standard normal transformations of that variable and the error term is zero. Hence, testing the lack of correlation between the standard normal transformations of that variable and the error term is a necessary and also sufficient condition for testing the exogeneity of instruments. 
\end{sloppypar}
Next, using the Gaussian copula as the joint distribution of $\epsilon_t$, $Z_t$ and $\eta_t$, we decompose the error term $\epsilon_t$ into two components. The first is a linear combination of the standard normal transformations of $Z_t$ and $\eta_t$, to capture the endogeneity of $P_t$. The second component is a normal, independent error term essential for the consistent estimation of the regression model. We add the normal transformations of $Z_t$ and $\eta_t$ as added regressors to the original regression equation and estimate the model. We establish that the test for the exogeneity of the instruments is equivalent to a Wald test, where a linear transformation for the coefficients of the standard normal transformations of the instruments $Z_t$ should be equal to zero. 

Next, by simplifying the reduced-form equation for $P_t$ and setting $\eta_t=P_t$, we establish an instrument-free approach to test for the exogeneity of a regressor. This is a valuable approach given that in the existing literature, the regressor endogeneity test boils down to comparing the coefficient of the regressor under OLS and TSLS or testing if the coefficient of the error term from the first stage is not statistically significant in the second stage (\cite{durbin1954errors}, \cite{hausman1978specification}, \cite{wu1973alternative}, \cite{wooldridge2010econometric}). However, both of these methods rely on having an exogenous instrument from the error term. In other words, the traditional approach to test for the endogeneity of a regressor is not helpful if the researcher does not have access to an instrument. In addition, this approach can lead to a wrong conclusion if the instrument used in the test is actually not exogenous to the error term. In particular, if the instruments used in the Hausman endogeneity test are not exogenous to the error term, this test will reject the exogeneity hypothesis even if the regressor is actually exogenous. In a recent work, \cite{caetano2015test} develops an instrument-free approach to test the endogeneity of regressors in models with bunching. However, their method is based on two assumptions: 1) the outcome variable is continuous in the endogenous regressor, and 2) there is at least one unobservable confounder that is discontinuously distributed with respect to the endogenous variable. These two assumptions may not hold in various situations. Hence, the applicability of the method in \cite{caetano2015test} could be limited. 

Through a series of simulation studies, we evaluated our test's efficacy, yielding the following key findings:
\begin{itemize}
   \item \textbf{Instrument Exogeniety Test:} In cases where instruments are truly exogenous, our test dismisses the exogeneity assumption on average 5.8\% of the times for $T=200$ and 7\% of the times for $T=1,000$ at the 5\% significance level. Conversely, if instruments are actually endogenous, our test rejects the exogeneity hypothesis on average 76.6\% of the times for $T=200$ and 98.8\% of the times for $T=1,000$ at the 5\% significance level. 
   
   In other words, the type-I error of our test is less than 10\% and the type-II error is less than 30\%, which can be decreased to 3\% for larger data sets. Notably, our test remains robust even when the error term deviates from a normal distribution.
   
   \item \textbf{Regressor Exogeneity Test:} Our method demonstrates superior accuracy compared to the Hausman test for endogeneity, even when an exogenous instrument is available for analysis. Importantly, while the Hausman test performance significantly deteriorates if the instrument used in the test is not exogenous, our test's effectiveness remains unaffected due to its instrument-free nature.
\end{itemize}

A shortcoming of our instrument exogeneity test is that if an instrument is normally distributed and highly correlated with the endogenous variable $P_t$, its normal transformation will be highly correlated with $P_t$ creating the multicollinearity issue. Moreover, similar to \cite{park2012handling}, if the endogenous variable $P_t$ has a normal distribution, its normal transformation will be $P_t$ divided by its standard deviation, creating the multicollinearity issue in the regressor endogeneity test. In these cases, collecting a data set with a higher number of observations can mitigate the efficiency loss caused by multicollinearity.

Furthermore, we apply our test empirically in a Two-Stage Least Squares (TSLS) setting in Angrist's study on compulsory education, \cite{angristcompulsory}. Our test successfully validates the endogeneity of the education variable. Additionally, it confirms the exogeneity of the instruments used in this study, which include interactions between the year of birth and the quarter of birth. We also show that after adding the demographic variables as regressors the endogeneity of the education variable disappears. This empirical application underscores the practical effectiveness of our test in real-world TSLS scenarios.

To the best of our knowledge, our work is the first one that provides a rigorous approach to test instruments and regressor exogeneity with minimal assumptions and straightforward implementation. Our method holds significant promise for practitioners who previously had to depend on untestable arguments to substantiate their methodologies and findings. We are confident that our work will greatly improve the rigor and reliability in the field of IV regression, offering substantial benefits to those in applied econometrics and related disciplines.
\section{Testing the Exogeneity of Instrumental Variables Using Copulas}
We aim to estimate the following linear regression model:

\begin{equation}\label{reg}
    Y_t=\beta_0+\beta^T X_t+\alpha P_t+\epsilon_t,~t=\{1,...,T\},
\end{equation}
where $X_t$ represents a $k\times 1$ vector of exogenous regressors and $P_t$ is a scalar endogenous regressor. The term $\epsilon_t$ is the error component of the model. Here, $t$ indexes either time or cross-sectional units. We have at our disposal $m$ instrumental variables excluded from the original regression model to address the endogeneity of $P_t$. Denote by $Z_{t}$ the $(m\times 1)$ vector of instrumental variables, where $Z_{it}$ is the $i$th instrument. It is hypothesized that these instrumental variables may exhibit correlation with the error term $\epsilon_t$. Furthermore, for any variable $\kappa_t$, we denote its marginal CDF, its mean, and its standard deviation by $F_{\kappa}(\cdot)$, $\mu_{\kappa}$ and $\sigma_{\kappa}$, respectively. 

Consider the reduced-form equation for the endogenous variable expressed as follows:
\begin{equation}\label{firststagem}
    P_t=\delta_0+\delta^T X_t+\gamma^T Z_{t}+\eta_t.
\end{equation}
Given the correlation between $P_t$ and the error term in the main regression equation, and assuming the exogeneity of $X_t$, the potential correlation of $Z_t$ and $\eta_t$ with the error term $\epsilon_t$ becomes the focal point of consideration. We operate under the assumption that $Z_t$ and $\eta_t$ (the error term in the reduced-form equation) are uncorrelated.

According to Sklar's Theorem (\cite{sklar1959fonctions}), any multivariate joint distribution can be expressed by a copula function alongside the marginal distributions of the involved variables. Based on this theorem, we model the endogeneity of $Z_t$ and $\eta_t$ by assuming that a Gaussian copula can represent the joint distribution of $Z_t$, $\eta_t$ and $\epsilon_t$.
\vspace{3mm}
\begin{assumption}\label{assum1}
    \textbf{A Gaussian copula represents the joint distribution of $Z_t$, $\eta_t$ and $\epsilon_t$.}
\end{assumption}
\vspace{3mm}
Suppose that $H(Z_t, \eta_t,\epsilon_t)$ denotes the joint Cumulative Distribution Function (CDF) of the instrumental variables, the error term in the reduced-form equation for the endogenous variable, and the error term in the original regression equation. We can express this relationship as:
\begin{equation}
   H(Z_t, \eta_t,\epsilon_t) = \Psi(Z_t^* = [Z_{it}^*]_{i=1, \ldots, m}, \eta_t^*, \epsilon_t^*), \label{copula1}
\end{equation}
where $\Psi$ represents the CDF of a standard multivariate normal distribution. The variables $Z_t^* = [Z_{it}^*]_{i=1, \ldots, m}$, $\eta_t^*$, and $\epsilon_t^*$ are the standard normal transformations of $Z_t$, $\eta_t$, and $\epsilon_t$, respectively. For any variable $\kappa_t \in \{Z_t, \eta_t, \epsilon_t\}$, its standard normal transformation $\kappa_t^*$ is defined as a standard normal random variable corresponding to the same CDF value. The standard normal transformation of a variable $\kappa_t$ follows a standard normal distribution. The normal transformation $\kappa_t^*$ is defined as follows:
\begin{itemize}
    \item For continuous $\kappa_t$: We define $\kappa_t^* = \Phi^{-1}(F_{\kappa}(\kappa_t))$, where $\Phi$ is the CDF of the standard normal distribution.
    \item For discrete $\kappa_t$: Assuming $\kappa_t$ takes on $n$ possible values $a_1 < a_2 < \ldots < a_n$ with $p(\kappa_t = a_i) = q_i$, and $F_{\kappa}(a_i) = \Sigma_{j=1}^{i} q_i$. For $\kappa_t = a_i$, we set $\kappa_t^* = \Phi^{-1}(u_i^*)$, where $u_i^*$ is a uniform random variable between $F_{\kappa}(a_i)$ and $F_{\kappa}(a_{i+1})$. We define $a_0 = -\infty$ and $a_{n+1} = \infty$. In the discrete case, $\kappa_t^*$ may vary for identical values of $\kappa_t = a_i$ since $u_i^*$ is a randomly drawn variable each time.
\end{itemize}

We represent the Pearson correlation coefficient between two variables, $\kappa_t$ and $\chi_t$, from the set $\{Z_t, \eta_t, \epsilon_t\}$ as $\rho_{\kappa\chi}$. Similarly, the Pearson correlation between their corresponding standard normal transformations, denoted by $\kappa_t^*$ and $\chi_t^*$, is expressed as $\rho_{\kappa^*\chi*}$.

We assume that the error term \(\epsilon_t\) has a normal distribution,  which is a widely used assumption in the literature (\cite{kleibergen2003bayesian}, \cite{ebbes2005solving}, \cite{rossi2003bayesian}, \cite{park2012handling} and \cite{yang2022addressing}). This assumption leads to the following relationships: \(F_{\epsilon}(\epsilon_t) = \Phi(\sigma_{\epsilon} \epsilon_t^*)\) and \(\epsilon = \sigma_{\epsilon} \epsilon_t^*\).


\vspace{3mm}
\begin{assumption}\label{assum2}
    \textbf{The error term \(\epsilon_t\) has a mean-zero normal distribution \(N(0,\sigma_{\epsilon}^2)\).}
\end{assumption}
\vspace{3mm}
In the subsequent analysis, we establish that for any variable $\kappa_t$, the condition $\rho_{\kappa \epsilon} = 0$ holds if and only if $\rho_{\kappa^* \epsilon^*} = 0$. Put differently, the exogeneity of $\kappa_t$ is synonymous with the absence of correlation between its standard normal transformation $\kappa_t^*$ and $\epsilon_t^* = \epsilon_t / \sigma_{\epsilon}$. This result will be the basis for our exogeneity test. 


We emphasize the non-triviality of this result, particularly in distinguishing between exogeneity and independence conditions. The former merely necessitates the absence of correlation with the error term, as opposed to the latter's requirement of independence. Consider a continuous variable $\kappa_t \in \{Z_t, \eta_t, \epsilon_t\}$. We have $\kappa_t^*=\Phi^{-1}(F_{\kappa}(\kappa_t))$, indicating that $\kappa_t^*$ is a monotone transformation of $\kappa_t$. Consequently, if $\kappa_t$ is independent of $\epsilon_t$, it directly follows that $\kappa_t^*$ is independent of the transformed error term $\epsilon_t^*=\epsilon_t/\sigma_{\epsilon}$, leading to a correlation $\rho_{\kappa^* \epsilon^*} = 0$. However, the mere absence of correlation, $\rho_{\kappa \epsilon} = 0$, does not imply independence. Moreover, the standard normal transformation of $\kappa_t$ given by $\Phi^{-1}(F_{\kappa}(\kappa_t))$ is a non-linear transformation. Therefore, $\rho_{\kappa \epsilon} = 0$ resulting in $\rho_{\kappa^* \epsilon^*} = 0$ is not a trivial argument.

This result will be demonstrated first for variables $\kappa_t$ that are continuous, followed by an examination of discrete $\kappa_t$. We begin by presenting a proposition that articulates the relationship between $\rho_{\kappa \epsilon}$ and $\rho_{\kappa^* \epsilon^*}$ for a continuous variable $\kappa_t$.
\vspace{3mm}
\begin{proposition}\label{prop1}
    \textbf{
    Let \(\epsilon_t\) be a continuous mean-zero normal random variable, and \(\kappa_t\) a continuous random variable with CDF, mean, and standard deviation denoted by \(F_{\kappa}(\cdot)\), \(\mu_{\kappa}\), and \(\sigma_{\kappa}\) respectively. The correlation between \(\kappa_t\) and \(\epsilon_t\), denoted by \(\rho_{\kappa\epsilon}\), is related to the correlation factor of their standard normal transformations, denoted by \(\rho_{\kappa^* \epsilon^*}\), as follows:
    \begin{equation}\label{correleq}
        \rho_{\kappa\epsilon}=\rho_{\kappa^* \epsilon^*}~\frac{\int \nu F_{\kappa}^{-1}(\Phi(\nu)) \phi(\nu) d\nu }{\sigma_{\kappa}}.
    \end{equation}
    Furthermore, the relationship is bounded as follows:
    \begin{equation}\label{correlbound}
0 < |\frac{\int \nu F_{\kappa}^{-1}(\Phi(\nu)) \phi(\nu) d\nu }{\sigma_{\kappa}}| \leq       {\sqrt{1+(\frac{\mu_{\kappa}}{\sigma_{\kappa}})^2}}
    \end{equation}
    }
\end{proposition}
\vspace{3mm}
\begin{proof}
We have 
\begin{equation}\label{peq1}
E(\kappa_t \epsilon_t)=\rho_{\kappa\epsilon}\sigma_{\kappa}\sigma_{\epsilon}+E(\kappa_t) E(\epsilon_t) .  
\end{equation}
On the other hand, we can write, 
\begin{align}\label{peq2}
E(\kappa_t \epsilon_t)=&\int \int \kappa_t \epsilon_t ~d\Psi(\Phi^{-1}(F_{\kappa}(\kappa_t)),\Phi^{-1}(F_{\epsilon}(\epsilon_t))),\\
&\int \int F_{\kappa}^{-1}(\Phi(\kappa_t^*)) F_{\epsilon}^{-1}(\Phi(\epsilon_t^*)) \phi(\kappa_t^*,\epsilon_t^*, \rho_{\kappa\epsilon*}) d\kappa_t^* d\epsilon_t^*,\
\end{align}    
where $\phi(\kappa_t^*,\epsilon_t^*, \rho_{\kappa\epsilon*})$ is the PDf of a bivariate normal distribution with the correlation factor $\rho_{\kappa^*\epsilon*}$. Using \eqref{peq1}, \eqref{peq2}, $E(\epsilon_t)=0$, and $F_{\epsilon}^{-1}(\Phi(\epsilon_t^*))=\sigma_{\epsilon}\Phi^{-1}(\Phi(\epsilon_t^*))=\sigma_{\epsilon} \epsilon^*$, we have
\begin{equation}\label{peq3}
    \rho_{\kappa\epsilon}\sigma_{\kappa}=\int F_{\kappa}^{-1}(\Phi(\kappa_t^*)) \Big( \int \epsilon_t^* \phi(\kappa_t^*,\epsilon_t^*, \rho_{\kappa\epsilon*}) d\epsilon_t^* \Big) d\kappa_t^*.
\end{equation}
From Theorem 5.10.4 in \cite{degroot2011probability}, we can simplify the inside integral in \eqref{peq3} as follows.
\begin{equation}\label{peq4}
 \int \epsilon_t^* \phi(\kappa_t^*,\epsilon_t^*, \rho_{\kappa\epsilon*}) d\epsilon_t^*=E(\epsilon_t^*|\kappa_t^*)=\kappa_t^* \rho_{\kappa^* \epsilon^*} \phi(\kappa_t^*).
\end{equation}
Using \eqref{peq3}-\eqref{peq4} and changing the variable $\kappa_t^*$ to $\nu$, we have
    \begin{equation}\label{mainst}
        \rho_{\kappa\epsilon}=\rho_{\kappa^* \epsilon^*}~\frac{\int \nu F_{\kappa}^{-1}(\Phi(\nu)) \phi(\nu) d\nu }{\sigma_{\kappa}}.
    \end{equation}
This proves the first part of the proposition. For the second part, we can rewrite $\int \nu F_{\kappa}^{-1}(\Phi(\nu)) \phi(\nu) d\nu$ as follows by changing the variable to $\lambda=\Phi(\nu)$:
\begin{equation}\label{intchange}
    \int \nu F_{\kappa}^{-1}(\Phi(\nu)) \phi(\nu) d\nu=\int_{0}^1 F_{\kappa}^{-1}(\lambda) \Phi^{-1}(\lambda)d\lambda.
\end{equation}
By the Cauchy-Schwarz inequality, we have 
\begin{equation}\label{csch}
    |\int_{0}^1F_{\kappa}^{-1}(\lambda) \Phi^{-1}(\lambda)d\lambda| \leq \sqrt{\int_0^1 (F_{\kappa}^{-1}(\lambda))^2 d\lambda}   \sqrt{\int_0^1 (\Phi^{-1}(\lambda))^2 d\lambda}.
\end{equation}
It is straightforward to show that $\int_0^1 (F_{\kappa}^{-1}(\lambda))^2 d\lambda=E(\kappa^2)=\sigma_{\kappa}^2+\mu_{\kappa}^2$, and $\int_0^1 (\Phi^{-1}(\lambda))^2 d\lambda=E_{\Phi}(\lambda^2)=1$. Hence, we can rewrite \eqref{csch} as follows.
\begin{equation}\label{csch2}
    |\int_{0}^1F_{\kappa}^{-1}(\lambda) \Phi^{-1}(\lambda)d\lambda| \leq \sqrt{\sigma_{\kappa}^2+\mu_{\kappa}^2}.
\end{equation}
Using \eqref{intchange} and \eqref{csch2}, we have
\begin{equation*}
|\frac{\int \nu F_{\kappa}^{-1}(\Phi(\nu)) \phi(\nu) d\nu }{\sigma_{\kappa}}| \leq       {\sqrt{1+(\frac{\mu_{\kappa}}{\sigma_{\kappa}})^2}},
\end{equation*}
which is the right side of the inequality in \ref{correlbound}. For the left side of the inequality in \ref{correlbound}, given that $\nu \phi(\nu)=-1 \times (-\nu) \phi(-\nu)$ and $\Phi(-\nu)=1-\phi(\nu)$, we have the following.
\begin{align}
    \int \nu F_{\kappa}^{-1}(\Phi(\nu)) \phi(\nu) d\nu &= \int_0^\infty \nu \phi(\nu) [F_{\kappa}^{-1}(\Phi(\nu))-F_{\kappa}^{-1}(\Phi(-\nu))] d\nu, \notag \\
    &=\int_0^\infty \nu \phi(\nu) [F_{\kappa}^{-1}(\Phi(\nu))-F_{\kappa}^{-1}(1-\Phi(\nu))] d\nu. \label{int2}
\end{align}
For $\nu>0$, we have $\nu \phi(\nu)>0$, and $\Phi(\nu)>1-\Phi(\nu)$. Given that $F_{\kappa}(\cdot)$ is a strictly increasing function, its inverse is a strictly increasing function as well. Therefore, for $\nu>0$ we have $F_{\kappa}^{-1}(\Phi(\nu))-F_{\kappa}^{-1}(1-\Phi(\nu)) >0$. Consequently, we have $\nu \phi(\nu) [F_{\kappa}^{-1}(\Phi(\nu))-F_{\kappa}^{-1}(1-\Phi(\nu))] >0$. This results in $\int \nu F_{\kappa}^{-1}(\Phi(\nu)) \phi(\nu) d\nu=\int_0^\infty \nu \phi(\nu) [F_{\kappa}^{-1}(\Phi(\nu))-F_{\kappa}^{-1}(1-\Phi(\nu))] d\nu >0$, which gives the left-hand side of the inequality in \ref{correlbound}.
This concludes the proof.
\end{proof}
The following proposition follows from Proposition \ref{prop1}.
\vspace{3mm}
\begin{proposition}\label{correlcoro}
\textbf{
Consider the scenario where \(\epsilon_t\) is a continuous mean-zero normal random variable, and \(\kappa_t\) is a continuous random variable. In this setting, the condition \(\rho_{\kappa \epsilon} = 0\) (indicating no correlation between \(\kappa_t\) and \(\epsilon_t\)) holds if and only if \(\rho_{\kappa^* \epsilon^*} = 0\) (implying no correlation between their standard normal transformations \(\kappa_t^*\) and \(\epsilon_t^*\)).
}
\end{proposition}
\begin{proof}
Proposition \ref{prop1} elucidates that for a continuous variable \(\kappa_t\), the correlation \(\rho_{\kappa \epsilon}\) can be expressed as \(\rho_{\kappa^* \epsilon^*} \times C\), where \(C\) is a non-zero bounded constant. Consequently, this implies that if \(\rho_{\kappa \epsilon} = 0\), then it necessarily follows that \(\rho_{\kappa^* \epsilon^*} = 0\), and conversely, if \(\rho_{\kappa^* \epsilon^*} = 0\), then \(\rho_{\kappa \epsilon} = 0\) as well.
\end{proof}

Proposition \ref{correldisceret} establishes the same result as Proposition \ref{correlcoro} for discrete $\kappa_t$.
\vspace{3mm}
\begin{proposition}\label{correldisceret}
\textbf{Consider \(\epsilon_t\) as a continuous mean-zero normal random variable and \(\kappa_t\) as a discrete random variable, where $\kappa_t$ takes on $n$ possible values $a_1 < a_2 < \ldots < a_n$ with $p(\kappa_t = a_i) = q_i$, and $F_{\kappa}(a_i) = \Sigma_{j=1}^{i} q_i$. The correlation \(\rho_{\kappa \epsilon} = 0\) if and only if \(\rho_{\kappa^* \epsilon^*} = 0\).}
\end{proposition}
\begin{proof}
For \(\kappa_t = a_i\), \(\kappa_t^* = \Phi^{-1}(u_i^*)\), where \(u_i^*\) is uniformly distributed between \(F_{\kappa}(a_i)\) and \(F_{\kappa}(a_{i+1})\), with \(a_0 = -\infty\) and \(a_{n+1} = \infty\). Knowing \(\kappa_t^*\) uniquely determines \(\kappa_t\), as distinct \(\kappa_t\) values lead to different \(\kappa_t^*\) values. Thus, \(E(\epsilon_t | \kappa_t^*, \kappa_t) = E(\epsilon_t | \kappa_t^*)\).

\(\kappa_t^*\) is derived based solely on \(\kappa_t = a_i\), making \(\epsilon_t\) independent of \(\kappa_t^*\) given \(\kappa_t = a_i\). This implies \(E(\epsilon_t | \kappa_t^*, \kappa_t) = E(\epsilon_t | \kappa_t)\). Hence, $E(\epsilon_t | \kappa_t)=E(\epsilon_t | \kappa_t^*)$. With \(\epsilon_t^* = \epsilon_t / \sigma_{\epsilon}\), it follows that $\sigma_{\epsilon}E(\epsilon_t^* | \kappa_t)=\sigma_{\epsilon}E(\epsilon_t^* | \kappa_t^*)$.

Suppose that $\rho_{\kappa^* \epsilon^*}=0$. Given that both $\epsilon_t^*$ and $\kappa_t^*$ are normal random variables, zero correlation means that they are independent. The independence of \(\epsilon_t^*\) and \(\kappa_t^*\) implies \(E(\epsilon_t^* | \kappa_t^*) = E(\epsilon_t^*) = 0\). As \(E(\epsilon_t | \kappa_t) = \sigma_{\epsilon}E(\epsilon_t^* | \kappa_t^*)\), it follows that \(E(\epsilon_t | \kappa_t) = 0\). Therefore, \(Cov(\epsilon_t, \kappa_t) = E(\epsilon_t \kappa_t) - E(\epsilon_t)E(\kappa_t) =E(E(\epsilon_t \kappa_t|\kappa_t))= E(\kappa_t E(\epsilon_t |\kappa_t))=0\), leading to \(\rho_{\kappa \epsilon} = 0\).

Conversely, suppose that $\rho_{\kappa \epsilon}=0$. Therefore, $E(\epsilon_t \kappa_t)-E(\epsilon_t )E(\kappa_t)=0$. Given $E(\epsilon_t)=0$, we have $E(\epsilon_t \kappa_t)=0$. We can write the following.
\begin{itemize}
    \item \(0 = E(\epsilon_t) = \Sigma_{i=1}^n q_i E(\epsilon_t | a_i)\). \textbf{(b1)}
    \item \(0 = E(\epsilon_t \kappa_t) = \Sigma_{i=1}^n q_i a_i E(\epsilon_t | a_i)\). \textbf{(b2)}
\end{itemize}
Given that $a_i$s are distinct numbers, \textbf{(b1)} and \textbf{(b2)} show that two weighted sums of $E(\epsilon_t|a_i)$s with different weights are zero. This can happen only if $E(\epsilon_t|a_i)=0$ for all $a_i$s. In other words, $E(\epsilon_t |\kappa_t)=E(\epsilon_t)=0$. This results in $E(\epsilon_t^* | \kappa_t^*)=E(\epsilon_t^*)=0$, which leads $E(\epsilon_t^* \kappa_t^*)=E(E(\kappa_t^* \epsilon_t^* | \kappa_t^*))=E(\kappa_t^* E( \epsilon_t^* \\ | \kappa_t^*))=0$. Therefore, we have $\rho_{\epsilon^* \kappa *}=0$.

\end{proof}
Building upon the findings of Propositions \ref{prop1}, \ref{correlcoro}, and \ref{correldisceret}, we formulate a methodology to test the exogeneity of the instrumental variables \(Z_t\). This test hinges on examining the correlation between the standard normal transformation of \(Z_t\), denoted as \(Z_t^*\), and the normalized error term \(\epsilon^*\). Specifically, the test seeks to ascertain whether \(\rho_{Z_i^* \epsilon^*} = 0\). The details and formal structure of this exogeneity test are outlined in Theorem \ref{zexo}.

\vspace{5mm}
\begin{theorem}\textbf{[Testing Exogeneity of the Instrumental Variables Using the Gaussian Copula] \label{zexo}Denote by $\Sigma_{Z^*}$ the correlation matrix for the standard normal transformation vector of the instrumental variables $Z_t^*$. Testing the exogeneity of the instruments $Z_t$ for the endogenous variable $P_t$ in $Y_t=\beta^T X_t+\alpha P_t+\epsilon_t$ is equivalent to testing for $\Sigma_{Z^*} \theta_{Z_t^*}=0$ in the following OLS regression
\begin{equation}\label{newform0}
    Y_t=\beta_0+\beta^T X_t+\alpha P_t+\theta_{Z_t^*}^T Z_t^*+\theta_{\eta_t^*} \eta_t^*+\xi_t,
\end{equation}
where $\eta_t^*$ is the standard normal transformation of the error term from the reduced-form equation of the endogenous variable given by
\begin{equation}\label{reducedformtheorem}
    P_t=\delta_0+\delta^T X_t+\gamma^T Z_{t}+\eta_t.
\end{equation}}
\end{theorem}
\begin{proof}
Under the Gaussian copula framework, the joint distribution of the standard normal transformations of the instrumental variables $Z_t^*$, the error term in the reduced-form equation for the endogenous variable $\eta_t^*$, and the error term in the original regression equation $\epsilon_t^*$
is characterized by a mean-zero multivariate normal distribution. This joint distribution encapsulates the interdependencies among these variables, as captured by the Gaussian copula, and is formalized as follows:
\begin{equation}\label{copulam}
\left[\begin{array}{l}
Z_t^{*} \\
\eta_t^{*} \\
\epsilon_t^{*}
\end{array}\right]\sim N\left([0]_{(m+2) \times 1},\Sigma_*\right).
\end{equation}\label{copulamsigma}
The covariance matrix $\Sigma_*$ is given by
\begin{equation}
    \Sigma_*=\left[\begin{array}{ccc}
\Sigma_{Z^*} & [0]_{(m \times 1)}  & \rho_{Z^* \epsilon^*}^T \\
{[0]}_{(1 \times m)} & 1 & \rho_{\eta^* \epsilon^*} \\
\rho_{Z^* \epsilon^*} & \rho_{\eta^* \epsilon^*} & 1
\end{array}\right],
\end{equation}
where $\Sigma_{Z^*}$ is the covariance matrix of $Z_{t}^*$, $\rho_{\eta^* \epsilon^*}$ is the correlation between $\eta_{t}^*$ and $\epsilon_t^*$, and $\rho_{Z^*\epsilon^*}=[\rho_{Z_{it}^* \epsilon_t^*}]_{(i=1 \ldots m)}$, where $\rho_{Z_{it}^* \epsilon_t^*}$ is the correlation between $Z_{it}^*$ and $\epsilon_t^*$.

Suppose that Cholesky decompositions of $\Sigma_*$ and $\Sigma_{Z^*}$ are given by $\Sigma_*=L L^T$ and $\Sigma_{Z^*}=L_Z L_Z^T$. We have the following.
\begin{equation}\label{choleskyL}
L=\left[\begin{array}{ccc}
L_{Z^*} & {[0]}_{(m \times 1)} & {[0]}_{(m \times 1)} \\
{[0]}_{(1 \times m)} & 1 & 0 \\
V_{(1 \times m)}& \tau & \zeta
\end{array}\right],  
\end{equation}
and 
\begin{equation}\label{choleskyLT}
L^{T}=\left[\begin{array}{ccc}
L_{Z^*}^T & {[0]}_{(m \times 1)} & V^{T} \\
{[0]}_{(1 \times m)} & 1 & \tau \\
{[0]}_{(1 \times m)} & 0 & \zeta
\end{array}\right].
\end{equation}
The relationships \(\tau = \rho_{\eta^* \epsilon^*}\) and \(\zeta = \sqrt{1 - \tau^2 - \sum_{i=1}^m V_i^2}\) can be readily established. Considering the Cholesky decomposition of \(\Sigma_*\), we obtain the following:
\begin{equation}\label{multinormal}
\left[\begin{array}{l}
Z_t^{*} \\
\eta_t^{*} \\
\epsilon_t^{*}
\end{array}\right]=L \times \left[\begin{array}{l}
\Omega \\
\omega_{m+1} \\
\omega_{m+2}
\end{array}\right] \text{where} \quad \Omega=\left[\omega_i\right]_{i=1\ldots m},
\end{equation}
and all $\omega_i$ are standard normal random variables. From \eqref{multinormal}, it is straightforward to see that $\omega_{m+1}=\eta_t^*$

Multiplying the third row of $L$ and the first column of $L^T$ gives us the vector $\rho_{Z^* \epsilon^*}$. Therefore, we have $V L_{Z^*}^T=\rho_{Z^* \epsilon *}$. Multiplying both sides by $(L_{Z^*}^T)^{-1}$, we have the following.
\begin{equation}\label{Vequation}
 V=\rho_{Z^* \epsilon^*} (L_{Z^*}^T)^{-1}.
\end{equation}
From \eqref{choleskyL} and \eqref{multinormal}, we can write $Z_t^*=L_{Z^*} \Omega$. Therefore, we have the following.
\begin{equation}\label{Omegaequation}
 \Omega=(L_{Z^*})^{-1} Z_t^*.
\end{equation}
Moreover, from \eqref{choleskyL} and \eqref{multinormal}, we can write
\begin{equation}
    \epsilon_t^*=V \Omega+\rho_{\eta^* \epsilon^*} \omega_{m+1}+\zeta \omega_{m+2}.
\end{equation}
Substituting $V$ and $\Omega$ from equations \eqref{Vequation} and \eqref{Omegaequation}, we can rewrite $\epsilon_t^*$ as follows
\begin{align}
    \epsilon_t^* &=\rho_{Z^* \epsilon^*} (L_{Z^*}^T)^{-1} (L_{Z^*})^{-1} Z_t^*+\rho_{\eta^* \epsilon^*} \eta_t^*+\zeta \omega_{m+2} \notag\\ 
    &=\rho_{Z^* \epsilon^*}(L_{Z^*}L_{Z^*}^T)^{-1}Z_t^*+\rho_{\eta^* \epsilon^*} \eta_t^*+\zeta \omega_{m+2} \label{epsilon*m}\\
    &=\rho_{Z^* \epsilon^*}(\Sigma_{Z^*})^{-1}Z_t^*+\rho_{\eta^* \epsilon^*} \eta_t^*+\zeta \omega_{m+2}.\notag
\end{align}
Given that $\epsilon_t=\sigma_{\epsilon} \epsilon_t^*$, using \eqref{epsilon*m}, we can rewrite the original OLS equation as follows:
\begin{equation}\label{newform1}
    Y_t=\beta^T X_t+\alpha P_t+\sigma_{\epsilon}\rho_{Z^* \epsilon^*}(\Sigma_{Z^*})^{-1}Z_t^*+\sigma_{\epsilon}\rho_{\eta^* \epsilon^*} \eta_t^*+\sigma_{\epsilon}\zeta \omega_{m+2}.
\end{equation}
Because $\omega_{m+2}$ is not correlated with any of the regressors in \eqref{newform1}, we can consistently estimate this OLS equation. Assuming $\xi_t=\sigma_{\epsilon}\zeta \omega_{m+2}$, $\theta_{\eta_t^*}=\sigma_{\epsilon}\rho_{\eta^* \epsilon^*}$, and $\theta_{Z_t^*}=\sigma_{\epsilon}(\Sigma_{Z^*})^{-1}\rho_{Z^* \epsilon^*}^T$, equations \eqref{newform0} and \eqref{newform1} are equivalent. Note that because $\Sigma_{Z^*}$ is a symmetric matrix, we have $\Sigma_{Z^*}^T=\Sigma_{Z^*}$.

From Corollary \ref{correlcoro}, we know that the exogeneity condition is equivalent to the test for $\rho_{Z^* \epsilon^*}=0$. Given that $\Sigma_{Z^*} \theta_{Z_t^*}= \sigma_{\epsilon}\rho_{Z^* \epsilon^*}$, testing for the exogeneity of the instruments is the same as testing for $\Sigma_{Z^*} \theta_{Z_t^*}=0.$ This can be done by a Wald test to verify the condition $\Sigma_{Z^*} \theta_{Z_t^*}=0.$.

\end{proof}
\subsection{Discussion}
Theorem \ref{zexo} demonstrates that testing the exogeneity of the instrumental variables essentially reduces to a Wald test for \(\Sigma_{Z^*} \theta_{Z_t^*} = 0\). Furthermore, the exogeneity of the error term in the reduced-form equation for the endogenous variable \(\eta_t^*\) can be assessed by verifying \(\theta_{\eta_t^*} = 0\). Notably, \(\theta_{\eta_t^*} = 0\) implies \(\rho_{\eta^* \epsilon^*} = 0\), which subsequently leads to \(\rho_{\eta \epsilon} = 0\).

A potential scenario that may adversely affect the performance of our exogeneity test involves multicollinearity in the estimation of \eqref{newform0}. Specifically, if \(P_t\) is highly correlated with \(Z_t^*\), it leads to multicollinearity, which in turn causes inefficiency in regression estimates and the Wald test. Such high correlation may emerge under two conditions:
\begin{enumerate}
    \item An instrumental variable \(Z_{it}\) has a strong correlation with \(P_t\).
    \item The instrument follows a normal distribution, resulting in \(Z_{it}^* = Z_{it} / \sigma_{Z_{it}}\).
\end{enumerate}
A larger sample size may be required to mitigate the inefficiency issue in scenarios satisfying these conditions,.

Furthermore, considering that \(\Sigma_{Z^*} \theta_{Z_t^*} = \sigma_{\epsilon}\rho_{Z^* \epsilon^*}\), we can utilize the estimated coefficients ($\widehat{\theta_{Z_t^*}}$) and the Root Mean Square Error (RMSE) from the regression equation ($\widehat{\sigma_{\epsilon}}$) in \eqref{newform0} to derive an estimate for \(\widehat{\rho_{Z^* \epsilon^*}} = [\widehat{\rho_{Z_i^* \epsilon^*}}]\). Subsequently for continuous instrumental variables, employing the results from Proposition \ref{prop1}, we are able to calculate an estimate \(\widehat{\rho_{Z \epsilon}} = [\widehat{\rho_{Z_i \epsilon}}]\). This approach enables us to gauge the extent of endogeneity for each instrument, thereby aiding in more informed instrument selection. Instruments exhibiting higher levels of correlation can be excluded, and we also gain insight into the potential bias introduced in the Two-Stage Least Squares (TSLS) approach. This methodology not only refines instrument selection but also enhances the overall reliability of the econometric analysis.

In the next section, we delineate the application of Theorem \ref{zexo} for testing the endogeneity of regressors in the absence of instrumental variables. This approach  leverages the theorem's results to test for endogeneity directly, bypassing the traditional reliance on instruments.
\section{An Instrument-Free Approach to Test for Regressor Exogeneity Using Copulas}\label{endogeneity}
Suppose that we want to estimate the regression model in \eqref{reg}, where \(P_t\) is a potentially endogenous variable. Traditional methods to test endogeneity, such as those proposed by \cite{durbin1954errors}, \cite{hausman1978specification}, \cite{wu1973alternative}, and \cite{wooldridge2010econometric}, require access to an exogenous instrument. In scenarios devoid of such instruments, these methods offer no reliable means to ascertain the exogeneity of \(P_t\). This section introduces an instrument-free approach, leveraging Theorem \ref{zexo}, to test the exogeneity of a regressor.

Theorem \ref{zexo} provides a framework to test the exogeneity of instruments. Since \(\theta_{\eta_t^*}\) in \eqref{newform0} equals \(\sigma_{\epsilon}\rho_{\eta^* \epsilon^*}\), and in light of Propositions \ref{correlcoro} and \ref{correldisceret}, testing the exogeneity of \(\eta_t\) is analogous to verifying \(\theta_{\eta_t^*} = 0\). To assess the exogeneity of \(P_t\), we modify the reduced-form equation in \eqref{reducedformtheorem} by setting \(P_t^* = \eta_t\). Given that \(P_t^*\) follows a standard normal distribution, \(\eta_t^* = P_t^*\). This allows us to apply the results of Theorem \ref{zexo} under the assumption that there are no instruments and \(\eta_t^* = P_t^*\). The ensuing corollary formalizes this test:

\vspace{3mm}
\begin{corollary}\textbf{[A Test for Regressor Exogeneity Using Copulas]
Testing the exogeneity of the regressor \(P_t\) in the regression model \(Y_t = \beta_0 + \beta^T X_t + \alpha P_t + \epsilon_t\) is equivalent to testing \(\theta_{P_t^*} = 0\) in the following OLS regression:
\begin{equation}\label{newformendo}
    Y_t = \beta_0 + \beta^T X_t + \alpha P_t + \theta_{P_t^*} P_t^* + \xi_t,
\end{equation}
where \(P_t^*\) is the standard normal distribution transformation of \(P_t\).}
\end{corollary}
\begin{proof}
Per Theorem \ref{zexo}, \(\theta_{P_t^*} = \sigma_{\epsilon} ~\rho_{P^* \epsilon^*}\). Propositions \ref{correlcoro} and \ref{correldisceret} establish that \(\rho_{P \epsilon} = 0 \Leftrightarrow \rho_{P \epsilon^*} = 0\). Hence, testing for the exogeneity of \(P_t\) (i.e., \(\rho_{P \epsilon} = 0\)) is equivalent to testing \(\theta_{P_t^*} = 0\) in \eqref{newformendo}.
\end{proof}

\section{Simulation Studies}\label{simstudies}
In this section, we conduct simulation studies to demonstrate the efficacy of our copula-based tests in assessing the exogeneity of both instruments and regressors.
\subsection{Performance of the Exogeneity Test for the Instrumental Variables}\label{sim1}
We examine a scenario where both \(X_t\) and \(P_t\) are \(1 \times 1\) vectors, with three available instrumental variables for the endogenous variable \(P_t\). The data-generating process adheres to the Gaussian copula model described in \ref{copula1}. Initially, \(Z_t^*\), \(\eta_t^*\), \(\epsilon_t^*\), and \(X_t^*\) are generated using a multivariate normal distribution with a correlation matrix described below. Next, for each variable \(\kappa\) in the set \(\{Z_t, \eta_t, \epsilon_t, X_t\}\), we define \(\kappa_t = F_{\kappa}^{-1}(\kappa_t^*)\), where the CDF function $F_{\kappa}(\cdot)$ for each variable corresponds to the distributions specified below . The simulation parameters are as follows:
\begin{itemize}
    \item Except for the first instrument \(Z_{1t}\), which follows a Student-t distribution with two degrees of freedom, and the error term that may have a non-normal distribution, the rest of the variables ($X_t, P_t, \eta_t, Z_{2t}, Z_{3t}$) follow a standard normal distribution.
    \item The following distributions are considered for the error term \(\epsilon_t\): $N(0,1)$, a student-t with $df=2$, a uniform distribution between -0.5 and 0.5, an exponential distribution with a mean of 1, and a Beta distribution with both shape and scale parameters equal to 0.5. 
    \item Two sample sizes are evaluated: \(T = 200\) and \(T = 1,000\).
    \item Given the exogeneity of $X_t$, we assume $\rho_{X^* \epsilon^*}=0$.
    \item Four scenarios are considered for the endogeneity of the instrumental variables:
    \begin{itemize}
        \item Scenario 1: \(\rho_{Z^*\epsilon^*} = [0.0, 0.0, 0.0]\).
        \item Scenario 2: \(\rho_{Z^*\epsilon^*} = [0.0, 0.5, 0.0]\).
        \item Scenario 3: \(\rho_{Z^*\epsilon^*} = [0.3, 0.5, 0.0]\).
        \item Scenario 4: \(\rho_{Z^*\epsilon^*} = [0.3, 0.5, 0.7]\).
    \end{itemize}
    \item The covariance matrix \(\Sigma_{Z^*}\) is defined as:
    \[\Sigma_{Z^*} = \left[\begin{array}{ccc}
    1 & 0.2 & 0.3 \\
    0.2 & 1 & 0.4 \\
    0.3 & 0.4 & 1
    \end{array}\right].\]
    \item A correlation of \(\rho_{\eta^* \epsilon^*} = 0.5\) is assumed.
    \item Correlation between the exogenous variable and the instruments is modeled as \(\rho_{X^* Z^*} = [0.2, 0.2, 0.2]\).
    \item The endogenous variable is set as \(P_t = 1 + 0.1 X_t + 0.1 Z_{1t} + 0.2 Z_{2t} + 0.3 Z_{3t} + \eta_t\).
    \item The dependent variable is modeled as \(Y_t = 1 + 0.3 X_t + P_t + \epsilon_t\).
\end{itemize}

For each scenario, we conduct 100 simulations and apply the Wald test for the correlation factor of each instrument as stipulated in Theorem \ref{zexo}. The proportion of instances where the Wald test rejects the null hypothesis (that the correlation between the instrument and error term is zero) is recorded. Additionally, the correlation between \(P_t\) and \(\epsilon_t\) is analyzed to assess the degree of endogeneity. The results are summarized in Tables \ref{zexotable1} and \ref{zexotable3} at the 5\% significance level for $T=200$ and $T=1,000$, respectively. We summarize the main insights as follows:
\begin{itemize}
    \item If the instrument is not correlated with the error term, our test rejects the exogeneity hypothesis on average 5.8\% of the times for $T=200$ and 7\% of the times for $T=1,000$ at the 5\% significance level. In other words, the type-I error of our test is less than 6\% for $T=200$ and 7\% for $T=1,000$.
    \item If the instrument is correlated with the error term, our test rejects the exogeneity hypothesis on average 76.6\% of the times for $T=200$ and 98.8\% of the times for $T=1,000$ at the 5\% significance level. In other words, the type-II error of our test (not rejecting the exogeneity hypothesis when the instruments are correlated with the error term) is less than 30\% for $T=200$ and less than 2\% for $T=1,000$.
    \item Even though our model is developed based on the normality assumption for the error term, our test provides robust performance for non-normal error terms. 
\end{itemize}
We present the simulation results for the 1\% significance level in Appendix \ref{zexoappendix}. The insights are similar to the case with the 5\% significance level. 

\begin{table}[H]
\scriptsize.
\begin{centering}
\begin{tabular}{|c|c|c|c|c|c|c|c|}
\hline
Distribution   of $\epsilon$     & $\rho_{Z_1^*   \epsilon^*}$ & $\rho_{Z_2^*   \epsilon^*}$ & $\rho_{Z_3^*   \epsilon^*}$ & $\rho_{P   \epsilon}$ & \begin{tabular}[c]{@{}c@{}}\% of times   \\ $\rho_{Z_1 \epsilon}=0$ \\ is rejected\end{tabular} & \begin{tabular}[c]{@{}c@{}}\% of times \\  $\rho_{Z_2  \epsilon}=0$\\  is rejected\end{tabular} & \begin{tabular}[c]{@{}c@{}}\% of times \\ $\rho_{Z_3 \epsilon}=0$\\  is rejected\end{tabular} \\ \hline
\multirow{4}{*}{$N(0,1)$}        & 0.00                        & 0.00                        & 0.00                        & 0.43                  & 7\%                                                                                             & 4\%                                                                                             & 6\%                                                                                           \\ \cline{2-8} 
                                 & 0.00                        & 0.50                        & 0.00                        & 0.50                  & 5\%                                                                                             & 81\%                                                                                            & 2\%                                                                                           \\ \cline{2-8} 
                                 & 0.30                        & 0.50                        & 0.00                        & 0.57                  & 62\%                                                                                            & 79\%                                                                                            & 5\%                                                                                           \\ \cline{2-8} 
                                 & 0.30                        & 0.50                        & 0.70                        & 0.74                  & 86\%                                                                                            & 96\%                                                                                            & 99\%                                                                                          \\ \hline
\multirow{4}{*}{$t(2)$}          & 0.00                        & 0.00                        & 0.00                        & 0.36                  & 5\%                                                                                             & 6\%                                                                                             & 8\%                                                                                           \\ \cline{2-8} 
                                 & 0.00                        & 0.50                        & 0.00                        & 0.43                  & 7\%                                                                                             & 56\%                                                                                            & 6\%                                                                                           \\ \cline{2-8} 
                                 & 0.30                        & 0.50                        & 0.00                        & 0.48                  & 51\%                                                                                            & 77\%                                                                                            & 15\%                                                                                          \\ \cline{2-8} 
                                 & 0.30                        & 0.50                        & 0.70                        & 0.63                  & 49\%                                                                                            & 71\%                                                                                            & 80\%                                                                                          \\ \hline
\multirow{4}{*}{$U(-0.5,0.5)$}   & 0.00                        & 0.00                        & 0.00                        & 0.42                  & 4\%                                                                                             & 8\%                                                                                             & 9\%                                                                                           \\ \cline{2-8} 
                                 & 0.00                        & 0.50                        & 0.00                        & 0.50                  & 5\%                                                                                             & 75\%                                                                                            & 10\%                                                                                          \\ \cline{2-8} 
                                 & 0.30                        & 0.50                        & 0.00                        & 0.56                  & 64\%                                                                                            & 83\%                                                                                            & 7\%                                                                                           \\ \cline{2-8} 
                                 & 0.30                        & 0.50                        & 0.70                        & 0.73                  & 80\%                                                                                            & 95\%                                                                                            & 97\%                                                                                          \\ \hline
\multirow{4}{*}{$EXP(1)$}        & 0.00                        & 0.00                        & 0.00                        & 0.39                  & 4\%                                                                                             & 5\%                                                                                             & 6\%                                                                                           \\ \cline{2-8} 
                                 & 0.00                        & 0.50                        & 0.00                        & 0.47                  & 4\%                                                                                             & 70\%                                                                                            & 0\%                                                                                           \\ \cline{2-8} 
                                 & 0.30                        & 0.50                        & 0.00                        & 0.51                  & 54\%                                                                                            & 75\%                                                                                            & 13\%                                                                                          \\ \cline{2-8} 
                                 & 0.30                        & 0.50                        & 0.70                        & 0.68                  & 67\%                                                                                            & 89\%                                                                                            & 94\%                                                                                          \\ \hline
\multirow{4}{*}{$BETA(0.5,0.5)$} & 0.00                        & 0.00                        & 0.00                        & 0.41                  & 3\%                                                                                             & 6\%                                                                                             & 3\%                                                                                           \\ \cline{2-8} 
                                 & 0.00                        & 0.50                        & 0.00                        & 0.49                  & 5\%                                                                                             & 72\%                                                                                            & 4\%                                                                                           \\ \cline{2-8} 
                                 & 0.30                        & 0.50                        & 0.00                        & 0.54                  & 49\%                                                                                            & 74\%                                                                                            & 1\%                                                                                           \\ \cline{2-8} 
                                 & 0.30                        & 0.50                        & 0.70                        & 0.71                  & 77\%                                                                                            & 95\%                                                                                            & 100\%                                                                                         \\ \hline
\end{tabular}
\caption{The results of the copula-based approach to test the exogeneity of instruments for different distributions of the error term ($T=200$, 5\% significance level). The last three columns on the right show the percentage of times out of 100 simulation trials that the exogeneity hypothesis is rejected. }
\label{zexotable1}
\end{centering}
\end{table}
\begin{table}[H]
\scriptsize.
\begin{centering}
\begin{tabular}{|c|c|c|c|c|c|c|c|}
\hline
Distribution   of $\epsilon$     & $\rho_{Z_1^*   \epsilon^*}$ & $\rho_{Z_2^*   \epsilon^*}$ & $\rho_{Z_3^*   \epsilon^*}$ & $\rho_{P   \epsilon}$ & \begin{tabular}[c]{@{}c@{}}\% of times   \\ $\rho_{Z_1 \epsilon}=0$ \\ is rejected\end{tabular} & \begin{tabular}[c]{@{}c@{}}\% of times \\  $\rho_{Z_2  \epsilon}=0$\\  is rejected\end{tabular} & \begin{tabular}[c]{@{}c@{}}\% of times \\ $\rho_{Z_3 \epsilon}=0$\\  is rejected\end{tabular} \\ \hline
\multirow{4}{*}{$N(0,1)$}        & 0.00                        & 0.00                        & 0.00                        & 0.42                  & 5\%                                                                                             & 5\%                                                                                             & 6\%                                                                                           \\ \cline{2-8} 
                                 & 0.00                        & 0.50                        & 0.00                        & 0.51                  & 10\%                                                                                            & 100\%                                                                                           & 5\%                                                                                           \\ \cline{2-8} 
                                 & 0.30                        & 0.50                        & 0.00                        & 0.57                  & 100\%                                                                                           & 100\%                                                                                           & 9\%                                                                                           \\ \cline{2-8} 
                                 & 0.30                        & 0.50                        & 0.70                        & 0.75                  & 100\%                                                                                           & 100\%                                                                                           & 100\%                                                                                         \\ \hline
\multirow{4}{*}{$t(2)$}          & 0.00                        & 0.00                        & 0.00                        & 0.33                  & 2\%                                                                                             & 0\%                                                                                             & 2\%                                                                                           \\ \cline{2-8} 
                                 & 0.00                        & 0.50                        & 0.00                        & 0.39                  & 2\%                                                                                             & 97\%                                                                                            & 2\%                                                                                           \\ \cline{2-8} 
                                 & 0.30                        & 0.50                        & 0.00                        & 0.46                  & 86\%                                                                                            & 97\%                                                                                            & 23\%                                                                                          \\ \cline{2-8} 
                                 & 0.30                        & 0.50                        & 0.70                        & 0.59                  & 92\%                                                                                            & 97\%                                                                                            & 98\%                                                                                          \\ \hline
\multirow{4}{*}{$U(-0.5,0.5)$}   & 0.00                        & 0.00                        & 0.00                        & 0.41                  & 6\%                                                                                             & 4\%                                                                                             & 6\%                                                                                           \\ \cline{2-8} 
                                 & 0.00                        & 0.50                        & 0.00                        & 0.49                  & 9\%                                                                                             & 100\%                                                                                           & 11\%                                                                                          \\ \cline{2-8} 
                                 & 0.30                        & 0.50                        & 0.00                        & 0.55                  & 100\%                                                                                           & 100\%                                                                                           & 15\%                                                                                          \\ \cline{2-8} 
                                 & 0.30                        & 0.50                        & 0.70                        & 0.73                  & 100\%                                                                                           & 100\%                                                                                           & 100\%                                                                                         \\ \hline
\multirow{4}{*}{$EXP(1)$}        & 0.00                        & 0.00                        & 0.00                        & 0.39                  & 4\%                                                                                             & 5\%                                                                                             & 6\%                                                                                           \\ \cline{2-8} 
                                 & 0.00                        & 0.50                        & 0.00                        & 0.46                  & 7\%                                                                                             & 100\%                                                                                           & 12\%                                                                                          \\ \cline{2-8} 
                                 & 0.30                        & 0.50                        & 0.00                        & 0.51                  & 97\%                                                                                            & 99\%                                                                                            & 14\%                                                                                          \\ \cline{2-8} 
                                 & 0.30                        & 0.50                        & 0.70                        & 0.67                  & 100\%                                                                                           & 100\%                                                                                           & 100\%                                                                                         \\ \hline
\multirow{4}{*}{$BETA(0.5,0.5)$} & 0.00                        & 0.00                        & 0.00                        & 0.41                  & 5\%                                                                                             & 5\%                                                                                             & 4\%                                                                                           \\ \cline{2-8} 
                                 & 0.00                        & 0.50                        & 0.00                        & 0.48                  & 8\%                                                                                             & 100\%                                                                                           & 4\%                                                                                           \\ \cline{2-8} 
                                 & 0.30                        & 0.50                        & 0.00                        & 0.54                  & 100\%                                                                                           & 100\%                                                                                           & 13\%                                                                                          \\ \cline{2-8} 
                                 & 0.30                        & 0.50                        & 0.70                        & 0.71                  & 100\%                                                                                           & 100\%                                                                                           & 100\%                                                                                         \\ \hline
\end{tabular}
\caption{The results of the copula-based approach to test the exogeneity of instruments for different distributions of the error term ($T=1,000$, 5\% significance level). The last three columns on the right show the percentage of times out of 100 simulation trials that the exogeneity hypothesis is rejected.}
\label{zexotable3}
\end{centering}
\end{table}


\subsection{Performance of the Exogeneity Test for the Regressors}\label{sim2}
To evaluate the performance of our test for regressor exogeneity, we generate the standard normal transformations of \(P_t\), \(X_t\), \(\epsilon_t\), and \(Z_t\) based on a specified correlation matrix described below using a multivariate normal distribution, subsequently creating the variables using their inverse CDF functions. We consider only one instrumental variable. This instrumental variable may have a potential correlation with the error term. We utilize this instrument to conduct the Hausman endogeneity test, comparing its outcomes with our instrument-free copula-based endogeneity test. The simulation setup is detailed as follows:
\begin{itemize}
  \item We assume that the endogenous variable $P_t$ has a student-t distribution with $df=2$.
  \item The exogenous variable \(X_t\) and the instrument \(Z_t\) follow standard normal distributions.
  \item The following distributions are considered for the error term \(\epsilon_t\): $N(0,1)$, a student-t with $df=2$, a uniform distribution between -0.5 and 0.5, an exponential distribution with mean of 1, and a beta distribution with both shape and scale parameters equal to 0.5. 
  \item Two scenarios are evaluated for the correlation between the instrument and the error term: \(\rho_{Z^* \epsilon^*} \in \{0, 0.2\}\).
  \item Nine different levels of endogeneity for \(P_t\) are considered: \(\rho_{P^* \epsilon^*} \in \{-0.5, -0.25, -0.1, -0.05, 0\\, 0.05, 0.1, 0.25, 0.5\}\).
  \item Correlations of \(\rho_{X^*P^*} = 0.2\) and \(\rho_{X^* Z^*} = 0.2\) are assumed to model potential correlations between the exogenous variable and the instrument. Given the exogeneity of $X_t$, we assume $\rho_{X^* \epsilon^*}=0$.
  \item The dependent variable is set as \(Y_t = 1 + 0.3 X_t + P_t + \epsilon_t\).
\end{itemize}
Tables \ref{tex01} and \ref{tex03} show the regressor endogeneity results for the case with $\rho_{Z^* \epsilon^*}=0$ at the 5\% significance level for $T=200$ and $T=1,000$, respectively, for various scenarios using both the instrument-free copula-based method and the Hausman approach. For brevity, we present the results for the case with $\rho_{Z^* \epsilon^*}=0$ at the 1\% significance level for $T=200$ and $T=1,000$, and the case with $\rho_{Z^* \epsilon^*}=0.2$ at the 5\% significance level for $T=1,000$ in Appendix \ref{appendixregressor}.
\begin{table}[]
\scriptsize.
\begin{centering}
\begin{tabular}{|c|c|c|c|c|}
\hline
Distribution   of $\epsilon$   & $\rho_{P^* \epsilon^*}$ & \begin{tabular}[c]{@{}c@{}}Average  $\rho_{P \epsilon}$ \\  in the data\end{tabular} & \begin{tabular}[c]{@{}c@{}}\% of times $\rho_{P \epsilon=0}$ \\  rejected, Copula method\end{tabular} & \begin{tabular}[c]{@{}c@{}}\% of times  $\rho_{P \epsilon=0}$  \\ rejected, Hausman method\end{tabular} \\ \hline
\multirow{9}{*}{$N(0,1)$}      & -0.5                & -0.43                                                                           & 90\%                                                                                          & 85\%                                                                                           \\ \cline{2-5} 
                               & -0.25               & -0.22                                                                           & 31\%                                                                                          & 32\%                                                                                           \\ \cline{2-5} 
                               & -0.1                & -0.08                                                                           & 10\%                                                                                          & 14\%                                                                                           \\ \cline{2-5} 
                               & -0.05               & -0.03                                                                           & 2\%                                                                                           & 7\%                                                                                            \\ \cline{2-5} 
                               & 0                   & 0.00                                                                            & 3\%                                                                                           & 2\%                                                                                            \\ \cline{2-5} 
                               & 0.05                & 0.04                                                                            & 7\%                                                                                           & 6\%                                                                                            \\ \cline{2-5} 
                               & 0.1                 & 0.08                                                                            & 8\%                                                                                           & 12\%                                                                                           \\ \cline{2-5} 
                               & 0.25                & 0.21                                                                            & 28\%                                                                                          & 37\%                                                                                           \\ \cline{2-5} 
                               & 0.5                 & 0.41                                                                            & 87\%                                                                                          & 75\%                                                                                           \\ \hline
\multirow{9}{*}{t(2)}        & -0.5                & -0.40                                                                           & 50\%                                                                                          & 76\%                                                                                           \\ \cline{2-5} 
                               & -0.25               & -0.19                                                                           & 23\%                                                                                          & 26\%                                                                                           \\ \cline{2-5} 
                               & -0.1                & -0.07                                                                           & 6\%                                                                                           & 6\%                                                                                            \\ \cline{2-5} 
                               & -0.05               & -0.04                                                                           & 1\%                                                                                           & 4\%                                                                                            \\ \cline{2-5} 
                               & 0                   & 0.01                                                                            & 1\%                                                                                           & 3\%                                                                                            \\ \cline{2-5} 
                               & 0.05                & 0.02                                                                            & 2\%                                                                                           & 2\%                                                                                            \\ \cline{2-5} 
                               & 0.1                 & 0.09                                                                            & 4\%                                                                                           & 9\%                                                                                            \\ \cline{2-5} 
                               & 0.25                & 0.18                                                                            & 26\%                                                                                          & 25\%                                                                                           \\ \cline{2-5} 
                               & 0.5                 & 0.40                                                                            & 50\%                                                                                          & 73\%                                                                                           \\ \hline
\multirow{9}{*}{U(-0.5,0.5)}   & -0.5                & -0.40                                                                           & 94\%                                                                                          & 75\%                                                                                           \\ \cline{2-5} 
                               & -0.25               & -0.21                                                                           & 38\%                                                                                          & 25\%                                                                                           \\ \cline{2-5} 
                               & -0.1                & -0.08                                                                           & 9\%                                                                                           & 13\%                                                                                           \\ \cline{2-5} 
                               & -0.05               & -0.04                                                                           & 6\%                                                                                           & 10\%                                                                                           \\ \cline{2-5} 
                               & 0                   & -0.01                                                                           & 2\%                                                                                           & 3\%                                                                                            \\ \cline{2-5} 
                               & 0.05                & 0.05                                                                            & 6\%                                                                                           & 6\%                                                                                            \\ \cline{2-5} 
                               & 0.1                 & 0.08                                                                            & 6\%                                                                                           & 8\%                                                                                            \\ \cline{2-5} 
                               & 0.25                & 0.20                                                                            & 36\%                                                                                          & 38\%                                                                                           \\ \cline{2-5} 
                               & 0.5                 & 0.41                                                                            & 93\%                                                                                          & 80\%                                                                                           \\ \hline
\multirow{9}{*}{EXP(1)}        & -0.5                & -0.39                                                                           & 67\%                                                                                          & 74\%                                                                                           \\ \cline{2-5} 
                               & -0.25               & -0.19                                                                           & 28\%                                                                                          & 23\%                                                                                           \\ \cline{2-5} 
                               & -0.1                & -0.06                                                                           & 2\%                                                                                           & 8\%                                                                                            \\ \cline{2-5} 
                               & -0.05               & -0.04                                                                           & 3\%                                                                                           & 1\%                                                                                            \\ \cline{2-5} 
                               & 0                   & 0.00                                                                            & 5\%                                                                                           & 7\%                                                                                            \\ \cline{2-5} 
                               & 0.05                & 0.04                                                                            & 1\%                                                                                           & 9\%                                                                                            \\ \cline{2-5} 
                               & 0.1                 & 0.07                                                                            & 9\%                                                                                           & 6\%                                                                                            \\ \cline{2-5} 
                               & 0.25                & 0.21                                                                            & 21\%                                                                                          & 39\%                                                                                           \\ \cline{2-5} 
                               & 0.5                 & 0.39                                                                            & 68\%                                                                                          & 72\%                                                                                           \\ \hline
\multirow{9}{*}{BETA(0.5,0.5)} & -0.5                & -0.39                                                                           & 94\%                                                                                          & 77\%                                                                                           \\ \cline{2-5} 
                               & -0.25               & -0.21                                                                           & 27\%                                                                                          & 32\%                                                                                           \\ \cline{2-5} 
                               & -0.1                & -0.10                                                                           & 7\%                                                                                           & 14\%                                                                                           \\ \cline{2-5} 
                               & -0.05               & -0.03                                                                           & 3\%                                                                                           & 3\%                                                                                            \\ \cline{2-5} 
                               & 0                   & 0.00                                                                            & 2\%                                                                                           & 4\%                                                                                            \\ \cline{2-5} 
                               & 0.05                & 0.02                                                                            & 4\%                                                                                           & 4\%                                                                                            \\ \cline{2-5} 
                               & 0.1                 & 0.09                                                                            & 4\%                                                                                           & 9\%                                                                                            \\ \cline{2-5} 
                               & 0.25                & 0.20                                                                            & 39\%                                                                                          & 26\%                                                                                           \\ \cline{2-5} 
                               & 0.5                 & 0.38                                                                            & 94\%                                                                                          & 74\%                                                                                           \\ \hline
\end{tabular}
\caption{The regressor exogeneity test results for $T=200$ and $\rho_{Z^* \epsilon^*}=0$ at the 5\% significance level. The last two columns on the right show the percentage of times out of 100 simulation trials that the exogeneity hypothesis is rejected.}
\label{tex01}
\end{centering}
\end{table}
\begin{table}[]
\scriptsize.
\begin{centering}
\begin{tabular}{|c|c|c|c|c|}
\hline
Distribution   of $\epsilon$   & $\rho_{P^* \epsilon^*}$ & \begin{tabular}[c]{@{}c@{}}Average  $\rho_{P \epsilon}$ \\  in the data\end{tabular} & \begin{tabular}[c]{@{}c@{}}\% of times $\rho_{P \epsilon=0}$ \\  rejected, Copula method\end{tabular} & \begin{tabular}[c]{@{}c@{}}\% of times  $\rho_{P \epsilon=0}$  \\ rejected, Hausman method\end{tabular} \\ \hline
\multirow{9}{*}{$N(0,1)$}      & -0.5                & -0.40                                                                           & 100\%                                                                                         & 96\%                                                                                           \\ \cline{2-5} 
                               & -0.25               & -0.20                                                                           & 98\%                                                                                          & 73\%                                                                                           \\ \cline{2-5} 
                               & -0.1                & -0.08                                                                           & 35\%                                                                                          & 22\%                                                                                           \\ \cline{2-5} 
                               & -0.05               & -0.04                                                                           & 11\%                                                                                          & 15\%                                                                                           \\ \cline{2-5} 
                               & 0                   & 0.00                                                                            & 1\%                                                                                           & 1\%                                                                                            \\ \cline{2-5} 
                               & 0.05                & 0.04                                                                            & 16\%                                                                                          & 14\%                                                                                           \\ \cline{2-5} 
                               & 0.1                 & 0.08                                                                            & 47\%                                                                                          & 19\%                                                                                           \\ \cline{2-5} 
                               & 0.25                & 0.20                                                                            & 97\%                                                                                          & 78\%                                                                                           \\ \cline{2-5} 
                               & 0.5                 & 0.40                                                                            & 100\%                                                                                         & 96\%                                                                                           \\ \hline
\multirow{9}{*}{t(2)}        & -0.5                & -0.35                                                                           & 81\%                                                                                          & 88\%                                                                                           \\ \cline{2-5} 
                               & -0.25               & -0.15                                                                           & 82\%                                                                                          & 50\%                                                                                           \\ \cline{2-5} 
                               & -0.1                & -0.06                                                                           & 39\%                                                                                          & 8\%                                                                                            \\ \cline{2-5} 
                               & -0.05               & -0.03                                                                           & 15\%                                                                                          & 10\%                                                                                           \\ \cline{2-5} 
                               & 0                   & 0.00                                                                            & 2\%                                                                                           & 3\%                                                                                            \\ \cline{2-5} 
                               & 0.05                & 0.03                                                                            & 13\%                                                                                          & 3\%                                                                                            \\ \cline{2-5} 
                               & 0.1                 & 0.06                                                                            & 28\%                                                                                          & 16\%                                                                                           \\ \cline{2-5} 
                               & 0.25                & 0.17                                                                            & 77\%                                                                                          & 62\%                                                                                           \\ \cline{2-5} 
                               & 0.5                 & 0.33                                                                            & 80\%                                                                                          & 83\%                                                                                           \\ \hline
\multirow{9}{*}{U(-0.5,0.5)}   & -0.5                & -0.37                                                                           & 100\%                                                                                         & 90\%                                                                                           \\ \cline{2-5} 
                               & -0.25               & -0.19                                                                           & 100\%                                                                                         & 71\%                                                                                           \\ \cline{2-5} 
                               & -0.1                & -0.07                                                                           & 45\%                                                                                          & 12\%                                                                                           \\ \cline{2-5} 
                               & -0.05               & -0.04                                                                           & 19\%                                                                                          & 7\%                                                                                            \\ \cline{2-5} 
                               & 0                   & 0.00                                                                            & 8\%                                                                                           & 5\%                                                                                            \\ \cline{2-5} 
                               & 0.05                & 0.04                                                                            & 16\%                                                                                          & 10\%                                                                                           \\ \cline{2-5} 
                               & 0.1                 & 0.07                                                                            & 47\%                                                                                          & 13\%                                                                                           \\ \cline{2-5} 
                               & 0.25                & 0.18                                                                            & 99\%                                                                                          & 66\%                                                                                           \\ \cline{2-5} 
                               & 0.5                 & 0.36                                                                            & 100\%                                                                                         & 84\%                                                                                           \\ \hline
\multirow{9}{*}{EXP(1)}        & -0.5                & -0.35                                                                           & 100\%                                                                                         & 85\%                                                                                           \\ \cline{2-5} 
                               & -0.25               & -0.18                                                                           & 96\%                                                                                          & 72\%                                                                                           \\ \cline{2-5} 
                               & -0.1                & -0.07                                                                           & 40\%                                                                                          & 24\%                                                                                           \\ \cline{2-5} 
                               & -0.05               & -0.04                                                                           & 7\%                                                                                           & 5\%                                                                                            \\ \cline{2-5} 
                               & 0                   & 0.00                                                                            & 2\%                                                                                           & 4\%                                                                                            \\ \cline{2-5} 
                               & 0.05                & 0.03                                                                            & 15\%                                                                                          & 7\%                                                                                            \\ \cline{2-5} 
                               & 0.1                 & 0.07                                                                            & 37\%                                                                                          & 20\%                                                                                           \\ \cline{2-5} 
                               & 0.25                & 0.18                                                                            & 94\%                                                                                          & 58\%                                                                                           \\ \cline{2-5} 
                               & 0.5                 & 0.36                                                                            & 99\%                                                                                          & 90\%                                                                                           \\ \hline
\multirow{9}{*}{BETA(0.5,0.5)} & -0.5                & -0.35                                                                           & 100\%                                                                                         & 90\%                                                                                           \\ \cline{2-5} 
                               & -0.25               & -0.18                                                                           & 100\%                                                                                         & 73\%                                                                                           \\ \cline{2-5} 
                               & -0.1                & -0.07                                                                           & 41\%                                                                                          & 17\%                                                                                           \\ \cline{2-5} 
                               & -0.05               & -0.03                                                                           & 9\%                                                                                           & 5\%                                                                                            \\ \cline{2-5} 
                               & 0                   & 0.00                                                                            & 3\%                                                                                           & 0\%                                                                                            \\ \cline{2-5} 
                               & 0.05                & 0.04                                                                            & 13\%                                                                                          & 9\%                                                                                            \\ \cline{2-5} 
                               & 0.1                 & 0.08                                                                            & 38\%                                                                                          & 17\%                                                                                           \\ \cline{2-5} 
                               & 0.25                & 0.19                                                                            & 98\%                                                                                          & 71\%                                                                                           \\ \cline{2-5} 
                               & 0.5                 & 0.35                                                                            & 100\%                                                                                         & 86\%                                                                                           \\ \hline
\end{tabular}
\caption{The regressor exogeneity test results for $T=1,000$ and $\rho_{Z^* \epsilon^*}=0$ at the 5\% significance level. The last two columns on the right show the percentage of times out of 100 simulation trials that the exogeneity hypothesis is rejected.}
\label{tex03}
\end{centering}
\end{table}
Key insights drawn from these simulation results are summarized below:
\begin{itemize}
    \item Relative to the traditional Hausman approach, the copula-based method exhibits greater accuracy. Specifically, it shows a higher frequency of correctly rejecting the exogeneity hypothesis when the regressor is not exogenous and a similar or lower frequency of incorrectly rejecting the hypothesis when the regressor is exogenous.
    \item Both the copula-based and Hausman approaches demonstrate significantly improved performance with an increased number of observations (\(T = 200\) versus \(T = 1,000\)).
    \item Our test shows a robust performance even for non-normal error terms.
    \item As illustrated in Table \ref{tex07} in Appendix \ref{appendixregressor}, the effectiveness of the Hausman approach diminishes in scenarios where the instrument correlates with the error term, often erroneously rejecting the exogeneity hypothesis even for exogenous regressors. Conversely, as the copula-based approach is not reliant on instruments, its performance remains unaffected by the endogeneity of the instrument.
   \item The type-I error of our approach (rejecting the exogeneity hypothesis when the regressor is exogenous) is less than 3\% for both $T=200$ and $T=1,000$.
   \item The type-II error of our approach (not rejecting the exogeneity when the regressor is endogenous) depends on the degree of endogeneity as follows:
   \begin{itemize}
       \item For $|\rho_{Z^* \epsilon^*}=0.5|$, it is on average less than 21\% for $T=200$ and less than 5\% for $T=1,000$.
       \item For $|\rho_{Z^* \epsilon^*}=0.25|$, it is on average less than 70\% for $T=200$ and less than 6\% for $T=1,000$.
       \item For $0<|\rho_{Z^* \epsilon^*}| \leq 0.1$, it is on average less than 93\% for $T=200$ and less than 74\% for $T=1,000$.
   \end{itemize}
\end{itemize}
We emphasize that even though the type-II error of the copula based test is high for low levels of endogeneity, still our copula based method outperforms the Hausman endogeneity test approach. 

\section{Empirical Application: Instrumental Variables in Education and Earnings}

This section evaluates the performance of our exogeneity tests in the context of Angrist's seminal paper on instrumental variables (\cite{angristcompulsory}). The paper examines the impact of education duration on earnings, addressing the endogeneity issue arising from the unobserved ability of a student influencing both education duration and wage. The authors employ a two-stage-least-square (TSLS) method to assess the effect of education (\(E_i\)) on the log of wage (\(\ln W_i\)), using quarter of birth dummies interacted with year of birth dummies as instruments.

We apply our exogeneity tests to the endogenous variable (education length, \(EDUC\)) and the 30 instrumental variables (\(QTR120-QTR129\), \(QTR220-QTR229\), and \(QTR320-QTR329\)) for men born between 1920 and 1929. Our analysis mirrors the TSLS results presented in Table IV of \cite{angristcompulsory}, specifically for Cases 2, 4, 6, and 8, described as follows:
\begin{itemize}
    \item Case 2: \(X_i\) includes year of birth.
    \item Case 4: \(X_i\) includes year of birth, age, and age-squared.
    \item Case 6: \(X_i\) includes year of birth, race dummies, a dummy for residence in SMSA, a marital status dummy, and eight region-of-residence dummies.
    \item Case 8: \(X_i\) includes year of birth, age and age-squared, race dummies, a dummy for residence in SMSA, a marital status dummy, and eight region-of-residence dummies.
\end{itemize}
Considering the discrete nature of \(EDUC\) and the instrumental variables, and the randomness in the normal distribution transformation for discrete variables, we execute our test 100 times with different random draws for the normal distribution transformations. We then calculate the frequency, out of 100 trials, at which the exogeneity hypothesis is rejected at the 5\% significance level. The results are shown in Table \ref{ang1}. We provide the results for the 1\% significance level in Appendix \ref{angappendix}.

\begin{table}[h]
\scriptsize.
\begin{centering}
\begin{tabular}{|c|c|c|c|c|}
\hline
\% of   times the exogeneity hypothesis is rejected at 0.05 & Case 2 & Case 4 & Case 6 & Case 8 \\ \hline
EDUC                                                        & 77\%   & 80\%   & 6\%    & 2\%    \\ \hline
QTR120                                                      & 1\%    & 3\%    & 5\%    & 3\%    \\ \hline
QTR121                                                      & 7\%    & 3\%    & 7\%    & 5\%    \\ \hline
QTR122                                                      & 4\%    & 5\%    & 11\%   & 2\%    \\ \hline
QTR123                                                      & 3\%    & 4\%    & 6\%    & 5\%    \\ \hline
QTR124                                                      & 4\%    & 4\%    & 2\%    & 1\%    \\ \hline
QTR125                                                      & 1\%    & 6\%    & 4\%    & 4\%    \\ \hline
QTR126                                                      & 0\%    & 3\%    & 8\%    & 5\%    \\ \hline
QTR127                                                      & 7\%    & 1\%    & 10\%   & 4\%    \\ \hline
QTR128                                                      & 4\%    & 8\%    & 4\%    & 11\%   \\ \hline
QTR129                                                      & 20\%   & 8\%    & 15\%   & 11\%   \\ \hline
QTR220                                                      & 4\%    & 2\%    & 6\%    & 8\%    \\ \hline
QTR221                                                      & 5\%    & 4\%    & 3\%    & 3\%    \\ \hline
QTR222                                                      & 2\%    & 3\%    & 6\%    & 4\%    \\ \hline
QTR223                                                      & 2\%    & 4\%    & 5\%    & 6\%    \\ \hline
QTR224                                                      & 1\%    & 6\%    & 4\%    & 6\%    \\ \hline
QTR225                                                      & 9\%    & 2\%    & 10\%   & 7\%    \\ \hline
QTR226                                                      & 6\%    & 4\%    & 2\%    & 9\%    \\ \hline
QTR227                                                      & 4\%    & 5\%    & 8\%    & 6\%    \\ \hline
QTR228                                                      & 3\%    & 4\%    & 4\%    & 7\%    \\ \hline
QTR229                                                      & 3\%    & 2\%    & 5\%    & 6\%    \\ \hline
QTR320                                                      & 4\%    & 1\%    & 7\%    & 2\%    \\ \hline
QTR321                                                      & 4\%    & 5\%    & 4\%    & 10\%   \\ \hline
QTR322                                                      & 11\%   & 5\%    & 7\%    & 9\%    \\ \hline
QTR323                                                      & 4\%    & 6\%    & 7\%    & 3\%    \\ \hline
QTR324                                                      & 3\%    & 5\%    & 3\%    & 5\%    \\ \hline
QTR325                                                      & 3\%    & 10\%   & 4\%    & 2\%    \\ \hline
QTR326                                                      & 4\%    & 2\%    & 9\%    & 7\%    \\ \hline
QTR327                                                      & 2\%    & 3\%    & 3\%    & 2\%    \\ \hline
QTR328                                                      & 6\%    & 5\%    & 9\%    & 6\%    \\ \hline
QTR329                                                      & 5\%    & 3\%    & 5\%    & 8\%    \\ \hline
\end{tabular}
\caption{The results of the exogeneity test for the endogenous variable and the instruments in \cite{angristcompulsory} at the 5\% significance level. The table shows the percentage of times out of 100 simulation trials that the exogeneity hypothesis is rejected.}
\label{ang1}
\end{centering}
\end{table}
As evidenced in Table \ref{ang1}, for cases 2 and 4, the exogeneity hypothesis of the endogenous variable \(EDUC\) is rejected in over 77\% of the simulations, strongly indicating the endogeneity of the education length variable. However, this apparent endogeneity measure diminishes to 6\% and 2\% in cases 6 and 8, where demographic variables are included, suggesting that these demographic factors account for the portions of the error term correlated with \(EDUC\).

Furthermore, Table \ref{ang1} indicates that for all instrumental variables across all cases, the exogeneity hypothesis is rejected on average 5\% of the time. This consistently low rejection rate provides robust evidence supporting the exogeneity of the instruments.
\section{Conclusion}
In conclusion, our study addresses a pivotal challenge in causal inference using instrumental variables – the testing of the exogeneity condition, which has long been a contentious and unresolved issue in applied econometrics. We propose a novel Copula-based method, a significant departure from traditional practices that largely relied on economic-theoretical justifications or untestable assumptions. Our approach, grounded in modeling the joint distribution of error terms and variables using a Gaussian copula, marks a substantial advancement in verifying the exogeneity of both instruments and regressors.

Our findings, derived from extensive simulation studies, demonstrate the robustness and accuracy of our test. Notably, the instrument exogeneity test shows high efficacy, correctly rejecting exogenous or endogenous instruments in most cases, even under non-normal error term distributions. Furthermore, our regressor exogeneity test outperforms the Hausman test, providing more reliable results without the need for exogenous instruments.

The empirical application of our test in a TSLS setting, using Angrist's study on compulsory education, further validates its practical utility. We successfully confirmed the endogeneity of the education variable sing an instrument-free approach. We also confirmed and the exogeneity of the birth year and quarter interaction instruments, exemplifying the test's applicability in real-world scenarios.

Our work makes a significant contribution to the field of applied econometrics by introducing a reliable and easy-to-implement approach for testing exogeneity. This advancement enhances the rigor and reliability of econometric analyses and offers substantial benefits to researchers and practitioners in the field. It opens new avenues for conducting more accurate and credible empirical research, paving the way for informed econometric analysis across various applied contexts. 
\bibliographystyle{apalike}
\bibliography{sample}

\appendix
\section{Additional Simulation Results for the Instrument Exogeneity Test}\label{zexoappendix}

\begin{table}[H]\label{zexotable4}
\scriptsize.
\begin{centering}
\begin{tabular}{|c|c|c|c|c|c|c|c|}
\hline
Distribution   of $\epsilon$     & $\rho_{Z_1^*   \epsilon^*}$ & $\rho_{Z_2^*   \epsilon^*}$ & $\rho_{Z_3^*   \epsilon^*}$ & $\rho_{P   \epsilon}$ & \begin{tabular}[c]{@{}c@{}}\% of times   \\ $\rho_{Z_1 \epsilon}=0$ \\ is rejected\end{tabular} & \begin{tabular}[c]{@{}c@{}}\% of times \\  $\rho_{Z_2  \epsilon}=0$\\  is rejected\end{tabular} & \begin{tabular}[c]{@{}c@{}}\% of times \\ $\rho_{Z_3 \epsilon}=0$\\  is rejected\end{tabular} \\ \hline
\multirow{4}{*}{$N(0,1)$}        & 0.00                        & 0.00                        & 0.00                        & 0.43                  & 0\%                                                                                             & 1\%                                                                                             & 1\%                                                                                           \\ \cline{2-8} 
                                 & 0.00                        & 0.50                        & 0.00                        & 0.50                  & 0\%                                                                                             & 66\%                                                                                            & 1\%                                                                                           \\ \cline{2-8} 
                                 & 0.30                        & 0.50                        & 0.00                        & 0.57                  & 45\%                                                                                            & 65\%                                                                                            & 2\%                                                                                           \\ \cline{2-8} 
                                 & 0.30                        & 0.50                        & 0.70                        & 0.74                  & 76\%                                                                                            & 89\%                                                                                            & 96\%                                                                                          \\ \hline
\multirow{4}{*}{$t(2)$}          & 0.00                        & 0.00                        & 0.00                        & 0.36                  & 2\%                                                                                             & 2\%                                                                                             & 2\%                                                                                           \\ \cline{2-8} 
                                 & 0.00                        & 0.50                        & 0.00                        & 0.43                  & 1\%                                                                                             & 38\%                                                                                            & 2\%                                                                                           \\ \cline{2-8} 
                                 & 0.30                        & 0.50                        & 0.00                        & 0.48                  & 30\%                                                                                            & 68\%                                                                                            & 12\%                                                                                          \\ \cline{2-8} 
                                 & 0.30                        & 0.50                        & 0.70                        & 0.63                  & 36\%                                                                                            & 61\%                                                                                            & 66\%                                                                                          \\ \hline
\multirow{4}{*}{$U(-0.5,0.5)$}   & 0.00                        & 0.00                        & 0.00                        & 0.42                  & 0\%                                                                                             & 0\%                                                                                             & 0\%                                                                                           \\ \cline{2-8} 
                                 & 0.00                        & 0.50                        & 0.00                        & 0.50                  & 2\%                                                                                             & 66\%                                                                                            & 3\%                                                                                           \\ \cline{2-8} 
                                 & 0.30                        & 0.50                        & 0.00                        & 0.56                  & 44\%                                                                                            & 73\%                                                                                            & 1\%                                                                                           \\ \cline{2-8} 
                                 & 0.30                        & 0.50                        & 0.70                        & 0.73                  & 65\%                                                                                            & 88\%                                                                                            & 93\%                                                                                          \\ \hline
\multirow{4}{*}{$EXP(1)$}        & 0.00                        & 0.00                        & 0.00                        & 0.39                  & 2\%                                                                                             & 2\%                                                                                             & 3\%                                                                                           \\ \cline{2-8} 
                                 & 0.00                        & 0.50                        & 0.00                        & 0.47                  & 0\%                                                                                             & 62\%                                                                                            & 0\%                                                                                           \\ \cline{2-8} 
                                 & 0.30                        & 0.50                        & 0.00                        & 0.51                  & 35\%                                                                                            & 56\%                                                                                            & 7\%                                                                                           \\ \cline{2-8} 
                                 & 0.30                        & 0.50                        & 0.70                        & 0.68                  & 43\%                                                                                            & 77\%                                                                                            & 84\%                                                                                          \\ \hline
\multirow{4}{*}{$BETA(0.5,0.5)$} & 0.00                        & 0.00                        & 0.00                        & 0.41                  & 0\%                                                                                             & 1\%                                                                                             & 0\%                                                                                           \\ \cline{2-8} 
                                 & 0.00                        & 0.50                        & 0.00                        & 0.49                  & 1\%                                                                                             & 60\%                                                                                            & 1\%                                                                                           \\ \cline{2-8} 
                                 & 0.30                        & 0.50                        & 0.00                        & 0.54                  & 32\%                                                                                            & 63\%                                                                                            & 1\%                                                                                           \\ \cline{2-8} 
                                 & 0.30                        & 0.50                        & 0.70                        & 0.71                  & 62\%                                                                                            & 88\%                                                                                            & 92\%                                                                                          \\ \hline
\end{tabular}
\caption{The results of the copula-based approach to test the exogeneity of instruments for different distributions of the error term ($T=200$, 1\% significance level). The last three columns on the right show the percentage of times out of 100 simulation trials that the exogeneity hypothesis is rejected.}
\end{centering}
\end{table}
\begin{table}[H]\label{zexotable6}
\scriptsize.
\begin{centering}
\begin{tabular}{|c|c|c|c|c|c|c|c|}
\hline
Distribution   of $\epsilon$     & $\rho_{Z_1^*   \epsilon^*}$ & $\rho_{Z_2^*   \epsilon^*}$ & $\rho_{Z_3^*   \epsilon^*}$ & $\rho_{P   \epsilon}$ & \begin{tabular}[c]{@{}c@{}}\% of times   \\ $\rho_{Z_1 \epsilon}=0$ \\ is rejected\end{tabular} & \begin{tabular}[c]{@{}c@{}}\% of times \\  $\rho_{Z_2  \epsilon}=0$\\  is rejected\end{tabular} & \begin{tabular}[c]{@{}c@{}}\% of times \\ $\rho_{Z_3 \epsilon}=0$\\  is rejected\end{tabular} \\ \hline
\multirow{4}{*}{$N(0,1)$}        & 0.00                        & 0.00                        & 0.00                        & 0.42                  & 0\%                                                                                             & 0\%                                                                                             & 2\%                                                                                           \\ \cline{2-8} 
                                 & 0.00                        & 0.50                        & 0.00                        & 0.51                  & 3\%                                                                                             & 100\%                                                                                           & 3\%                                                                                           \\ \cline{2-8} 
                                 & 0.30                        & 0.50                        & 0.00                        & 0.57                  & 100\%                                                                                           & 100\%                                                                                           & 2\%                                                                                           \\ \cline{2-8} 
                                 & 0.30                        & 0.50                        & 0.70                        & 0.75                  & 100\%                                                                                           & 100\%                                                                                           & 100\%                                                                                         \\ \hline
\multirow{4}{*}{$t(2)$}          & 0.00                        & 0.00                        & 0.00                        & 0.33                  & 0\%                                                                                             & 0\%                                                                                             & 0\%                                                                                           \\ \cline{2-8} 
                                 & 0.00                        & 0.50                        & 0.00                        & 0.39                  & 1\%                                                                                             & 96\%                                                                                            & 1\%                                                                                           \\ \cline{2-8} 
                                 & 0.30                        & 0.50                        & 0.00                        & 0.46                  & 78\%                                                                                            & 96\%                                                                                            & 9\%                                                                                           \\ \cline{2-8} 
                                 & 0.30                        & 0.50                        & 0.70                        & 0.59                  & 83\%                                                                                            & 94\%                                                                                            & 97\%                                                                                          \\ \hline
\multirow{4}{*}{$U(-0.5,0.5)$}   & 0.00                        & 0.00                        & 0.00                        & 0.41                  & 0\%                                                                                             & 1\%                                                                                             & 0\%                                                                                           \\ \cline{2-8} 
                                 & 0.00                        & 0.50                        & 0.00                        & 0.49                  & 3\%                                                                                             & 100\%                                                                                           & 3\%                                                                                           \\ \cline{2-8} 
                                 & 0.30                        & 0.50                        & 0.00                        & 0.55                  & 100\%                                                                                           & 100\%                                                                                           & 8\%                                                                                           \\ \cline{2-8} 
                                 & 0.30                        & 0.50                        & 0.70                        & 0.73                  & 100\%                                                                                           & 100\%                                                                                           & 100\%                                                                                         \\ \hline
\multirow{4}{*}{$EXP(1)$}        & 0.00                        & 0.00                        & 0.00                        & 0.39                  & 0\%                                                                                             & 0\%                                                                                             & 1\%                                                                                           \\ \cline{2-8} 
                                 & 0.00                        & 0.50                        & 0.00                        & 0.46                  & 2\%                                                                                             & 100\%                                                                                           & 3\%                                                                                           \\ \cline{2-8} 
                                 & 0.30                        & 0.50                        & 0.00                        & 0.51                  & 94\%                                                                                            & 99\%                                                                                            & 6\%                                                                                           \\ \cline{2-8} 
                                 & 0.30                        & 0.50                        & 0.70                        & 0.67                  & 99\%                                                                                            & 100\%                                                                                           & 100\%                                                                                         \\ \hline
\multirow{4}{*}{$BETA(0.5,0.5)$} & 0.00                        & 0.00                        & 0.00                        & 0.41                  & 0\%                                                                                             & 0\%                                                                                             & 0\%                                                                                           \\ \cline{2-8} 
                                 & 0.00                        & 0.50                        & 0.00                        & 0.48                  & 1\%                                                                                             & 100\%                                                                                           & 0\%                                                                                           \\ \cline{2-8} 
                                 & 0.30                        & 0.50                        & 0.00                        & 0.54                  & 100\%                                                                                           & 100\%                                                                                           & 3\%                                                                                           \\ \cline{2-8} 
                                 & 0.30                        & 0.50                        & 0.70                        & 0.71                  & 100\%                                                                                           & 100\%                                                                                           & 100\%                                                                                         \\ \hline
\end{tabular}
\caption{The results of the copula-based approach to test the exogeneity of instruments for different distributions of the error term ($T=1,000$, 1\% significance level). The last three columns on the right show the percentage of times out of 100 simulation trials that the exogeneity hypothesis is rejected.}
\end{centering}
\end{table}

\section{Additional Simulation Results for the Regressor Endogeneity Test}\label{appendixregressor}
\begin{table}[H]
\scriptsize.
\begin{centering}
\begin{tabular}{|c|c|c|c|c|}
\hline
Distribution   of $\epsilon$   & $\rho_{P^* \epsilon^*}$ & \begin{tabular}[c]{@{}c@{}}Average  $\rho_{P \epsilon}$ \\  in the data\end{tabular} & \begin{tabular}[c]{@{}c@{}}\% of times $\rho_{P \epsilon=0}$ \\  rejected, Copula method\end{tabular} & \begin{tabular}[c]{@{}c@{}}\% of times  $\rho_{P \epsilon=0}$  \\ rejected, Hausman method\end{tabular} \\ \hline
\multirow{9}{*}{$N(0,1)$}      & -0.5                & -0.43                                                                           & 67\%                                                                                          & 70\%                                                                                           \\ \cline{2-5} 
                               & -0.25               & -0.22                                                                           & 21\%                                                                                          & 16\%                                                                                           \\ \cline{2-5} 
                               & -0.1                & -0.08                                                                           & 1\%                                                                                           & 3\%                                                                                            \\ \cline{2-5} 
                               & -0.05               & -0.03                                                                           & 1\%                                                                                           & 2\%                                                                                            \\ \cline{2-5} 
                               & 0                   & 0.00                                                                            & 2\%                                                                                           & 1\%                                                                                            \\ \cline{2-5} 
                               & 0.05                & 0.04                                                                            & 2\%                                                                                           & 0\%                                                                                            \\ \cline{2-5} 
                               & 0.1                 & 0.08                                                                            & 4\%                                                                                           & 4\%                                                                                            \\ \cline{2-5} 
                               & 0.25                & 0.21                                                                            & 13\%                                                                                          & 12\%                                                                                           \\ \cline{2-5} 
                               & 0.5                 & 0.41                                                                            & 73\%                                                                                          & 62\%                                                                                           \\ \hline
\multirow{9}{*}{t(2)}        & -0.5                & -0.40                                                                           & 27\%                                                                                          & 55\%                                                                                           \\ \cline{2-5} 
                               & -0.25               & -0.19                                                                           & 9\%                                                                                           & 10\%                                                                                           \\ \cline{2-5} 
                               & -0.1                & -0.07                                                                           & 0\%                                                                                           & 1\%                                                                                            \\ \cline{2-5} 
                               & -0.05               & -0.04                                                                           & 0\%                                                                                           & 0\%                                                                                            \\ \cline{2-5} 
                               & 0                   & 0.01                                                                            & 0\%                                                                                           & 0\%                                                                                            \\ \cline{2-5} 
                               & 0.05                & 0.02                                                                            & 0\%                                                                                           & 0\%                                                                                            \\ \cline{2-5} 
                               & 0.1                 & 0.09                                                                            & 0\%                                                                                           & 1\%                                                                                            \\ \cline{2-5} 
                               & 0.25                & 0.18                                                                            & 9\%                                                                                           & 13\%                                                                                           \\ \cline{2-5} 
                               & 0.5                 & 0.40                                                                            & 34\%                                                                                          & 56\%                                                                                           \\ \hline
\multirow{9}{*}{U(-0.5,0.5)}   & -0.5                & -0.40                                                                           & 88\%                                                                                          & 56\%                                                                                           \\ \cline{2-5} 
                               & -0.25               & -0.21                                                                           & 19\%                                                                                          & 13\%                                                                                           \\ \cline{2-5} 
                               & -0.1                & -0.08                                                                           & 4\%                                                                                           & 3\%                                                                                            \\ \cline{2-5} 
                               & -0.05               & -0.04                                                                           & 0\%                                                                                           & 3\%                                                                                            \\ \cline{2-5} 
                               & 0                   & -0.01                                                                           & 0\%                                                                                           & 1\%                                                                                            \\ \cline{2-5} 
                               & 0.05                & 0.05                                                                            & 1\%                                                                                           & 2\%                                                                                            \\ \cline{2-5} 
                               & 0.1                 & 0.08                                                                            & 2\%                                                                                           & 5\%                                                                                            \\ \cline{2-5} 
                               & 0.25                & 0.20                                                                            & 17\%                                                                                          & 19\%                                                                                           \\ \cline{2-5} 
                               & 0.5                 & 0.41                                                                            & 89\%                                                                                          & 58\%                                                                                           \\ \hline
\multirow{9}{*}{EXP(1)}        & -0.5                & -0.39                                                                           & 52\%                                                                                          & 58\%                                                                                           \\ \cline{2-5} 
                               & -0.25               & -0.19                                                                           & 15\%                                                                                          & 9\%                                                                                            \\ \cline{2-5} 
                               & -0.1                & -0.06                                                                           & 0\%                                                                                           & 2\%                                                                                            \\ \cline{2-5} 
                               & -0.05               & -0.04                                                                           & 0\%                                                                                           & 1\%                                                                                            \\ \cline{2-5} 
                               & 0                   & 0.00                                                                            & 2\%                                                                                           & 2\%                                                                                            \\ \cline{2-5} 
                               & 0.05                & 0.04                                                                            & 0\%                                                                                           & 1\%                                                                                            \\ \cline{2-5} 
                               & 0.1                 & 0.07                                                                            & 4\%                                                                                           & 1\%                                                                                            \\ \cline{2-5} 
                               & 0.25                & 0.21                                                                            & 7\%                                                                                           & 11\%                                                                                           \\ \cline{2-5} 
                               & 0.5                 & 0.39                                                                            & 49\%                                                                                          & 52\%                                                                                           \\ \hline
\multirow{9}{*}{BETA(0.5,0.5)} & -0.5                & -0.39                                                                           & 84\%                                                                                          & 62\%                                                                                           \\ \cline{2-5} 
                               & -0.25               & -0.21                                                                           & 13\%                                                                                          & 17\%                                                                                           \\ \cline{2-5} 
                               & -0.1                & -0.10                                                                           & 1\%                                                                                           & 5\%                                                                                            \\ \cline{2-5} 
                               & -0.05               & -0.03                                                                           & 0\%                                                                                           & 0\%                                                                                            \\ \cline{2-5} 
                               & 0                   & 0.00                                                                            & 0\%                                                                                           & 0\%                                                                                            \\ \cline{2-5} 
                               & 0.05                & 0.02                                                                            & 2\%                                                                                           & 1\%                                                                                            \\ \cline{2-5} 
                               & 0.1                 & 0.09                                                                            & 1\%                                                                                           & 3\%                                                                                            \\ \cline{2-5} 
                               & 0.25                & 0.20                                                                            & 19\%                                                                                          & 16\%                                                                                           \\ \cline{2-5} 
                               & 0.5                 & 0.38                                                                            & 85\%                                                                                          & 52\%                                                                                           \\ \hline
\end{tabular}
\caption{The regressor exogeneity test results for $T=200$ and $\rho_{Z^* \epsilon^*}=0$ at the 1\% significance level. The last two columns on the right show the percentage of times out of 100 simulation trials that the exogeneity hypothesis is rejected.}
\label{tex04}
\end{centering}
\end{table}
\begin{table}[H]
\scriptsize.
\begin{centering}
\begin{tabular}{|c|c|c|c|c|}
\hline
Distribution   of $\epsilon$   & $\rho_{P^* \epsilon^*}$ & \begin{tabular}[c]{@{}c@{}}Average  $\rho_{P \epsilon}$ \\  in the data\end{tabular} & \begin{tabular}[c]{@{}c@{}}\% of times $\rho_{P \epsilon=0}$ \\  rejected, Copula method\end{tabular} & \begin{tabular}[c]{@{}c@{}}\% of times  $\rho_{P \epsilon=0}$  \\ rejected, Hausman method\end{tabular} \\ \hline
\multirow{9}{*}{$N(0,1)$}      & -0.5                & -0.40                                                                           & 100\%                                                                                         & 90\%                                                                                           \\ \cline{2-5} 
                               & -0.25               & -0.20                                                                           & 94\%                                                                                          & 55\%                                                                                           \\ \cline{2-5} 
                               & -0.1                & -0.08                                                                           & 16\%                                                                                          & 3\%                                                                                            \\ \cline{2-5} 
                               & -0.05               & -0.04                                                                           & 3\%                                                                                           & 7\%                                                                                            \\ \cline{2-5} 
                               & 0                   & 0.00                                                                            & 0\%                                                                                           & 0\%                                                                                            \\ \cline{2-5} 
                               & 0.05                & 0.04                                                                            & 4\%                                                                                           & 5\%                                                                                            \\ \cline{2-5} 
                               & 0.1                 & 0.08                                                                            & 24\%                                                                                          & 4\%                                                                                            \\ \cline{2-5} 
                               & 0.25                & 0.20                                                                            & 92\%                                                                                          & 57\%                                                                                           \\ \cline{2-5} 
                               & 0.5                 & 0.40                                                                            & 100\%                                                                                         & 86\%                                                                                           \\ \hline
\multirow{9}{*}{t(2)}        & -0.5                & -0.35                                                                           & 74\%                                                                                          & 78\%                                                                                           \\ \cline{2-5} 
                               & -0.25               & -0.15                                                                           & 73\%                                                                                          & 35\%                                                                                           \\ \cline{2-5} 
                               & -0.1                & -0.06                                                                           & 21\%                                                                                          & 1\%                                                                                            \\ \cline{2-5} 
                               & -0.05               & -0.03                                                                           & 6\%                                                                                           & 2\%                                                                                            \\ \cline{2-5} 
                               & 0                   & 0.00                                                                            & 1\%                                                                                           & 1\%                                                                                            \\ \cline{2-5} 
                               & 0.05                & 0.03                                                                            & 1\%                                                                                           & 2\%                                                                                            \\ \cline{2-5} 
                               & 0.1                 & 0.06                                                                            & 10\%                                                                                          & 6\%                                                                                            \\ \cline{2-5} 
                               & 0.25                & 0.17                                                                            & 60\%                                                                                          & 40\%                                                                                           \\ \cline{2-5} 
                               & 0.5                 & 0.33                                                                            & 70\%                                                                                          & 71\%                                                                                           \\ \hline
\multirow{9}{*}{U(-0.5,0.5)}   & -0.5                & -0.35                                                                           & 74\%                                                                                          & 78\%                                                                                           \\ \cline{2-5} 
                               & -0.25               & -0.15                                                                           & 73\%                                                                                          & 35\%                                                                                           \\ \cline{2-5} 
                               & -0.1                & -0.06                                                                           & 21\%                                                                                          & 1\%                                                                                            \\ \cline{2-5} 
                               & -0.05               & -0.03                                                                           & 6\%                                                                                           & 2\%                                                                                            \\ \cline{2-5} 
                               & 0                   & 0.00                                                                            & 1\%                                                                                           & 1\%                                                                                            \\ \cline{2-5} 
                               & 0.05                & 0.03                                                                            & 1\%                                                                                           & 2\%                                                                                            \\ \cline{2-5} 
                               & 0.1                 & 0.06                                                                            & 10\%                                                                                          & 6\%                                                                                            \\ \cline{2-5} 
                               & 0.25                & 0.17                                                                            & 60\%                                                                                          & 40\%                                                                                           \\ \cline{2-5} 
                               & 0.5                 & 0.33                                                                            & 70\%                                                                                          & 71\%                                                                                           \\ \hline
\multirow{9}{*}{EXP(1)}        & -0.5                & -0.35                                                                           & 99\%                                                                                          & 82\%                                                                                           \\ \cline{2-5} 
                               & -0.25               & -0.18                                                                           & 87\%                                                                                          & 51\%                                                                                           \\ \cline{2-5} 
                               & -0.1                & -0.07                                                                           & 18\%                                                                                          & 8\%                                                                                            \\ \cline{2-5} 
                               & -0.05               & -0.04                                                                           & 2\%                                                                                           & 1\%                                                                                            \\ \cline{2-5} 
                               & 0                   & 0.00                                                                            & 0\%                                                                                           & 1\%                                                                                            \\ \cline{2-5} 
                               & 0.05                & 0.03                                                                            & 6\%                                                                                           & 1\%                                                                                            \\ \cline{2-5} 
                               & 0.1                 & 0.07                                                                            & 17\%                                                                                          & 5\%                                                                                            \\ \cline{2-5} 
                               & 0.25                & 0.18                                                                            & 85\%                                                                                          & 35\%                                                                                           \\ \cline{2-5} 
                               & 0.5                 & 0.36                                                                            & 99\%                                                                                          & 79\%                                                                                           \\ \hline
\multirow{9}{*}{BETA(0.5,0.5)} & -0.5                & -0.35                                                                           & 100\%                                                                                         & 79\%                                                                                           \\ \cline{2-5} 
                               & -0.25               & -0.18                                                                           & 98\%                                                                                          & 54\%                                                                                           \\ \cline{2-5} 
                               & -0.1                & -0.07                                                                           & 18\%                                                                                          & 3\%                                                                                            \\ \cline{2-5} 
                               & -0.05               & -0.03                                                                           & 2\%                                                                                           & 0\%                                                                                            \\ \cline{2-5} 
                               & 0                   & 0.00                                                                            & 1\%                                                                                           & 0\%                                                                                            \\ \cline{2-5} 
                               & 0.05                & 0.04                                                                            & 3\%                                                                                           & 3\%                                                                                            \\ \cline{2-5} 
                               & 0.1                 & 0.08                                                                            & 21\%                                                                                          & 4\%                                                                                            \\ \cline{2-5} 
                               & 0.25                & 0.19                                                                            & 94\%                                                                                          & 47\%                                                                                           \\ \cline{2-5} 
                               & 0.5                 & 0.35                                                                            & 100\%                                                                                         & 82\%                                                                                           \\ \hline
\end{tabular}
\caption{The regressor exogeneity test results for $T=1,000$ and $\rho_{Z^* \epsilon^*}=0$ at the 1\% significance level. The last two columns on the right show the percentage of times out of 100 simulation trials that the exogeneity hypothesis is rejected.}
\label{tex06}
\end{centering}
\end{table}

\begin{table}[H]
\scriptsize.
\begin{centering}
\begin{tabular}{|c|c|c|c|c|}
\hline
Distribution   of $\epsilon$   & $\rho_{P^* \epsilon^*}$ & \begin{tabular}[c]{@{}c@{}}Average  $\rho_{P \epsilon}$ \\  in the data\end{tabular} & \begin{tabular}[c]{@{}c@{}}\% of times $\rho_{P \epsilon=0}$ \\  rejected, Copula method\end{tabular} & \begin{tabular}[c]{@{}c@{}}\% of times  $\rho_{P \epsilon=0}$  \\ rejected, Hausman method\end{tabular} \\ \hline
\multirow{9}{*}{$N(0,1)$}      & -0.5                & -0.38                                                                           & 100\%                                                                                         & 100\%                                                                                          \\ \cline{2-5} 
                               & -0.25               & -0.19                                                                           & 100\%                                                                                         & 100\%                                                                                          \\ \cline{2-5} 
                               & -0.1                & -0.08                                                                           & 35\%                                                                                          & 100\%                                                                                          \\ \cline{2-5} 
                               & -0.05               & -0.04                                                                           & 17\%                                                                                          & 100\%                                                                                          \\ \cline{2-5} 
                               & 0                   & 0.00                                                                            & 4\%                                                                                           & 100\%                                                                                          \\ \cline{2-5} 
                               & 0.05                & 0.03                                                                            & 15\%                                                                                          & 100\%                                                                                          \\ \cline{2-5} 
                               & 0.1                 & 0.08                                                                            & 47\%                                                                                          & 100\%                                                                                          \\ \cline{2-5} 
                               & 0.25                & 0.20                                                                            & 98\%                                                                                          & 98\%                                                                                           \\ \cline{2-5} 
                               & 0.5                 & 0.39                                                                            & 100\%                                                                                         & 20\%                                                                                           \\ \hline
\multirow{9}{*}{t(2)}        & -0.5                & -0.34                                                                           & 86\%                                                                                          & 100\%                                                                                          \\ \cline{2-5} 
                               & -0.25               & -0.15                                                                           & 84\%                                                                                          & 98\%                                                                                           \\ \cline{2-5} 
                               & -0.1                & -0.06                                                                           & 31\%                                                                                          & 100\%                                                                                          \\ \cline{2-5} 
                               & -0.05               & -0.03                                                                           & 8\%                                                                                           & 100\%                                                                                          \\ \cline{2-5} 
                               & 0                   & 0.00                                                                            & 6\%                                                                                           & 98\%                                                                                           \\ \cline{2-5} 
                               & 0.05                & 0.03                                                                            & 9\%                                                                                           & 97\%                                                                                           \\ \cline{2-5} 
                               & 0.1                 & 0.06                                                                            & 27\%                                                                                          & 99\%                                                                                           \\ \cline{2-5} 
                               & 0.25                & 0.15                                                                            & 76\%                                                                                          & 86\%                                                                                           \\ \cline{2-5} 
                               & 0.5                 & 0.35                                                                            & 78\%                                                                                          & 5\%                                                                                            \\ \hline
\multirow{9}{*}{U(-0.5,0.5)}   & -0.5                & -0.38                                                                           & 100\%                                                                                         & 100\%                                                                                          \\ \cline{2-5} 
                               & -0.25               & -0.19                                                                           & 99\%                                                                                          & 100\%                                                                                          \\ \cline{2-5} 
                               & -0.1                & -0.07                                                                           & 40\%                                                                                          & 100\%                                                                                          \\ \cline{2-5} 
                               & -0.05               & -0.04                                                                           & 10\%                                                                                          & 100\%                                                                                          \\ \cline{2-5} 
                               & 0                   & 0.00                                                                            & 4\%                                                                                           & 100\%                                                                                          \\ \cline{2-5} 
                               & 0.05                & 0.04                                                                            & 16\%                                                                                          & 100\%                                                                                          \\ \cline{2-5} 
                               & 0.1                 & 0.07                                                                            & 43\%                                                                                          & 100\%                                                                                          \\ \cline{2-5} 
                               & 0.25                & 0.19                                                                            & 100\%                                                                                         & 100\%                                                                                          \\ \cline{2-5} 
                               & 0.5                 & 0.37                                                                            & 100\%                                                                                         & 19\%                                                                                           \\ \hline
\multirow{9}{*}{EXP(1)}        & -0.5                & -0.36                                                                           & 100\%                                                                                         & 100\%                                                                                          \\ \cline{2-5} 
                               & -0.25               & -0.17                                                                           & 94\%                                                                                          & 100\%                                                                                          \\ \cline{2-5} 
                               & -0.1                & -0.07                                                                           & 35\%                                                                                          & 100\%                                                                                          \\ \cline{2-5} 
                               & -0.05               & -0.03                                                                           & 12\%                                                                                          & 100\%                                                                                          \\ \cline{2-5} 
                               & 0                   & 0.00                                                                            & 4\%                                                                                           & 100\%                                                                                          \\ \cline{2-5} 
                               & 0.05                & 0.04                                                                            & 14\%                                                                                          & 100\%                                                                                          \\ \cline{2-5} 
                               & 0.1                 & 0.07                                                                            & 42\%                                                                                          & 100\%                                                                                          \\ \cline{2-5} 
                               & 0.25                & 0.18                                                                            & 92\%                                                                                          & 92\%                                                                                           \\ \cline{2-5} 
                               & 0.5                 & 0.36                                                                            & 99\%                                                                                          & 20\%                                                                                           \\ \hline
\multirow{9}{*}{BETA(0.5,0.5)} & -0.5                & -0.36                                                                           & 100\%                                                                                         & 100\%                                                                                          \\ \cline{2-5} 
                               & -0.25               & -0.18                                                                           & 99\%                                                                                          & 100\%                                                                                          \\ \cline{2-5} 
                               & -0.1                & -0.07                                                                           & 35\%                                                                                          & 100\%                                                                                          \\ \cline{2-5} 
                               & -0.05               & -0.04                                                                           & 9\%                                                                                           & 100\%                                                                                          \\ \cline{2-5} 
                               & 0                   & 0.00                                                                            & 3\%                                                                                           & 100\%                                                                                          \\ \cline{2-5} 
                               & 0.05                & 0.04                                                                            & 15\%                                                                                          & 100\%                                                                                          \\ \cline{2-5} 
                               & 0.1                 & 0.07                                                                            & 38\%                                                                                          & 100\%                                                                                          \\ \cline{2-5} 
                               & 0.25                & 0.19                                                                            & 99\%                                                                                          & 97\%                                                                                           \\ \cline{2-5} 
                               & 0.5                 & 0.35                                                                            & 100\%                                                                                         & 21\%                                                                                           \\ \hline
\end{tabular}
\caption{The regressor exogeneity test results for $T=1,000$ and $\rho_{Z^* \epsilon^*}=0.2$ at the 5\% significance level. The last two columns on the right show the percentage of times out of 100 simulation trials that the exogeneity hypothesis is rejected.}
\label{tex07}
\end{centering}
\end{table}


\section{Additional Results for the setting of \cite{angristcompulsory}}\label{angappendix}

\begin{table}[H]
\scriptsize.
\begin{centering}
\begin{tabular}{|c|c|c|c|c|}
\hline
\% of   times the exogeneity hypothesis is rejected at 0.01 & Case 2 & Case 4 & Case 6 & Case 8 \\ \hline
EDUC                                                        & 50\%   & 56\%   & 0\%    & 1\%    \\ \hline
QTR120                                                      & 0\%    & 0\%    & 2\%    & 0\%    \\ \hline
QTR121                                                      & 1\%    & 0\%    & 1\%    & 0\%    \\ \hline
QTR122                                                      & 1\%    & 2\%    & 2\%    & 0\%    \\ \hline
QTR123                                                      & 1\%    & 1\%    & 0\%    & 0\%    \\ \hline
QTR124                                                      & 0\%    & 0\%    & 0\%    & 0\%    \\ \hline
QTR125                                                      & 0\%    & 2\%    & 0\%    & 2\%    \\ \hline
QTR126                                                      & 0\%    & 0\%    & 2\%    & 1\%    \\ \hline
QTR127                                                      & 1\%    & 0\%    & 1\%    & 0\%    \\ \hline
QTR128                                                      & 1\%    & 2\%    & 1\%    & 2\%    \\ \hline
QTR129                                                      & 4\%    & 1\%    & 3\%    & 3\%    \\ \hline
QTR220                                                      & 0\%    & 1\%    & 1\%    & 2\%    \\ \hline
QTR221                                                      & 0\%    & 0\%    & 1\%    & 1\%    \\ \hline
QTR222                                                      & 1\%    & 0\%    & 2\%    & 1\%    \\ \hline
QTR223                                                      & 1\%    & 3\%    & 3\%    & 0\%    \\ \hline
QTR224                                                      & 0\%    & 1\%    & 2\%    & 1\%    \\ \hline
QTR225                                                      & 2\%    & 0\%    & 1\%    & 4\%    \\ \hline
QTR226                                                      & 1\%    & 0\%    & 0\%    & 3\%    \\ \hline
QTR227                                                      & 1\%    & 0\%    & 3\%    & 1\%    \\ \hline
QTR228                                                      & 1\%    & 0\%    & 3\%    & 0\%    \\ \hline
QTR229                                                      & 1\%    & 1\%    & 0\%    & 1\%    \\ \hline
QTR320                                                      & 0\%    & 1\%    & 1\%    & 1\%    \\ \hline
QTR321                                                      & 0\%    & 2\%    & 1\%    & 5\%    \\ \hline
QTR322                                                      & 4\%    & 2\%    & 1\%    & 2\%    \\ \hline
QTR323                                                      & 1\%    & 1\%    & 1\%    & 1\%    \\ \hline
QTR324                                                      & 2\%    & 1\%    & 0\%    & 1\%    \\ \hline
QTR325                                                      & 0\%    & 1\%    & 0\%    & 0\%    \\ \hline
QTR326                                                      & 1\%    & 0\%    & 0\%    & 1\%    \\ \hline
QTR327                                                      & 1\%    & 0\%    & 1\%    & 0\%    \\ \hline
QTR328                                                      & 1\%    & 3\%    & 1\%    & 0\%    \\ \hline
QTR329                                                      & 0\%    & 0\%    & 1\%    & 1\%    \\ \hline
\end{tabular}
\caption{The results of the exogeneity test for the endogenous variable and the instruments in \cite{angristcompulsory} at the 1\% significance level. The table shows the percentage of times out of 100 simulation trials that the exogeneity hypothesis is rejected.}
\label{ang2}
\end{centering}
\end{table}
\end{document}